\documentclass[11pt,a4paper,english]{article}
\usepackage[titletoc, title]{appendix}
\usepackage{amsmath}
\usepackage{amssymb}
\usepackage{amsfonts}
\usepackage{mathrsfs}
\usepackage{bm}
\usepackage{bbm}
\usepackage{array}
\usepackage{babel}
\usepackage{bbding}
\usepackage{color}
\usepackage[normal]{caption}
\usepackage{subcaption}
\usepackage{epsfig}
\usepackage{graphicx}
\usepackage{pdflscape}
\usepackage{lipsum}
\usepackage{psfrag}
\usepackage{proofapnd}
\usepackage{amsthm}
\usepackage{siunitx}
\usepackage[hyphens]{url}
\usepackage[round]{natbib}
\usepackage[table,xcdraw]{xcolor}
\usepackage{afterpage}
\usepackage{enumerate}

\usepackage{rotating}
\usepackage[margin=2cm]{geometry} 
\usepackage{latexsym}
\usepackage{float}
\usepackage{setspace}
\usepackage{slashbox}
\usepackage{multirow, tabularx, longtable, ragged2e, booktabs, siunitx}

\newcolumntype{Y}{>{\centering\arraybackslash}X}
\newcolumntype{C}{>{\Centering\arraybackslash}X}
\newcolumntype{s}{>{\hsize=.2\hsize \Centering\arraybackslash}X}

\usepackage{authblk}
\usepackage[breaklinks=true,pagebackref=true]{hyperref}
\usepackage{indentfirst} 

\definecolor{markergreen}{rgb}{0.6, 1.0, 0}
\definecolor{darkgreen}{rgb}{0, .5, 0}
\definecolor{darkred}{rgb}{.7,0,0}
\definecolor{markergreen}{rgb}{0.6, 1.0, 0}
\definecolor{darkgreen}{rgb}{0, .5, 0}
\definecolor{darkorange}{rgb}{1,0.3,0}
\definecolor{darkred}{RGB}{.7,0,0}
\definecolor{darkblue}{RGB}{0,29,245}
\definecolor{orange}{RGB}{239, 133, 54}
\definecolor{lightblue}{RGB}{59, 188, 175}

\definecolor{plt1}{RGB}{31, 119, 180}
\definecolor{plt2}{RGB}{255, 127, 14}
\definecolor{plt3}{RGB}{44, 160, 44}
\definecolor{plt4}{RGB}{214, 39, 40}
\definecolor{plt5}{RGB}{148, 103, 189}
\definecolor{plt6}{RGB}{140, 86, 75}
\definecolor{plt7}{RGB}{227, 119, 194}
\definecolor{plt8}{RGB}{127, 127, 127}
\definecolor{plt9}{RGB}{188, 189, 34}
\definecolor{plt10}{RGB}{23, 190, 207}

\usepackage{changes}
\definechangesauthor[color=darkorange, name={Natalie}]{natp}

\title{\LARGE \bf Jump risk premia in the presence of clustered jumps}
\author{
	\begin{tabular}[t]{ccc}
    Francis Liu\thanks{
			Department of Business and Economics, Berlin School of Economics and Law, Badensche Str. 52, 10825 Berlin, Germany.
      Blockchain Research Center, Humboldt-Universität zu Berlin, Germany.
      International Research Training Group 1792, Humboldt-Universität zu Berlin, Germany.
     E-mail: \texttt{francisliutfp@gmail.com}.}

  Natalie Packham\thanks{
    Department of Business and Economics, Berlin School of Economics and Law, Badensche Str. 52, 10825 Berlin, Germany.
        International Research Training Group 1792, Humboldt-Universität zu Berlin, Germany.
    E-mail: \texttt{packham@hwr-berlin.de}.}

    Artur Sepp\thanks{LGT Bank (Schweiz) AG
 E-mail: \texttt{artursepp@gmail.com}.}

	\end{tabular}
}

\newcommand{\E}{{\mathbb{E}}}
\providecommand{\R}{{\mathbb{R}}}

\newcommand{\p}{{\mathbb P}}

\newcommand{\e}{{\bf e}}
\newcommand{\q}{{\mathbb Q}}

\newcommand{\dd}{{\rm d}}

\ifx\prop\undefined
\newtheorem{prop}{Proposition}[section]
\fi

\begin{document}
\setcommentmarkup{{\color{authorcolor}{}\em [#1]}}

\newtheorem{lemma}{Lemma}
\newtheorem {proposition}[lemma]{Proposition}
\newtheorem {corollary}{Corollary}
\newtheorem {theorem}{Theorem}
\newtheorem{claim}[lemma]{Claim}
\newtheorem{remark}[lemma]{Remark}
\newtheorem{example}[lemma]{Example}
\newtheorem{fact}[lemma]{Fact}
\newtheorem{defn}[lemma]{Definition}
\newtheorem{exercise}{Exercise}[section]
\newtheorem{programming}[exercise]{Programming assignment}
\newtheorem{assumption}{Assumption}[section]

\pagenumbering{arabic}

\maketitle

\begin{abstract}
This paper presents an option pricing model that incorporates clustered jumps using a bivariate Hawkes process. The process captures both self- and cross-excitation of positive and negative jumps, enabling the model to generate return dynamics with asymmetric, time-varying skewness and to produce positive or negative implied volatility skews. This feature is especially relevant for assets such as cryptocurrencies, so-called ``meme'' stocks, G-7 currencies, and certain commodities, where implied volatility skews may change sign depending on prevailing sentiment. 
We introduce two additional parameters, namely the positive and negative jump premia, to model the market risk preferences for positive and negative jumps, inferred from options data. This enables the model to flexibly match observed skew dynamics. Using Bitcoin (BTC) options, we empirically demonstrate how inferred jump risk premia exhibit predictive power for both the cost of carry in BTC futures and the performance of delta-hedged option strategies.


\pagestyle{empty}
\end{abstract}


\section{Introduction}\label{sec:introduction}
Jumps are acknowledged as an important risk factor in asset pricing, derivatives pricing, and risk management. Although there is a large body of literature investigating jumps in financial markets for traditional assets, the emergence of option markets for cryptocurrencies and retail-driven trading activity in options on ``meme'' stocks provide a new environment for this field of study. In this paper, we focus on the dynamics and options market for cryptocurrencies and, specifically, for Bitcoin (BTC).

Empirically, $\mathbb{P}$-dynamics of BTC exhibit both large positive and negative jumps, while option prices under $\mathbb{Q}$-dynamics exhibit changes in implied volatility skew from negative (when demand for protective put options increases) to positive (when demand for upside leverage rises). 

Figure \ref{fig:dynamics} illustrates the dynamics of Bitcoin (BTC) daily returns, at-the-money (ATM) implied volatilities and call and put skews from April 2019 to May 2024. First, we observe extreme returns ranging from $-30\%$ to $20\%$, with significant clustering. The implied at-the-money volatility for options with maturity of one month and Black-Scholes delta of $0.5$ ranges from $40\%$ to $140\%$. The call skew is computed as the difference between implied volatility of a call option with $0.25$ Black-Scholes delta and implied volatility of a call with $0.5$ delta. The put skew is computed as the difference between implied volatility of a put option with $-0.25$ Black-Scholes delta and the implied volatility of a put with $-0.5$ delta. Positive call and put skews indicate the higher demand for calls and puts, respectively. We observe that on average both call and put skews are positive. In addition, call skews exceed put skews in most periods, indicating a higher demand for calls than for puts. This behaviour contradicts traditional patterns in options on stocks and stock indices where the call skew is negative, while the put skew is consistently positive.

\begin{figure}[h]
    \centering
    \caption{Subplot (A) shows Bitcoin daily returns from April 2019 to May 2024 with extreme volatility and pronounced clustering during market stress periods. Subplot shows ATM Volatility daily with average of $68.3\%$. Subplot (C) shows BTC 25-Delta Call/Put Skews with corresponding averages of $3.0\%$ and $1.7\%$ and with blue/pink regions corresponding to periods with relatively higher call/put skews respectively.}
    \includegraphics[width=1\linewidth]{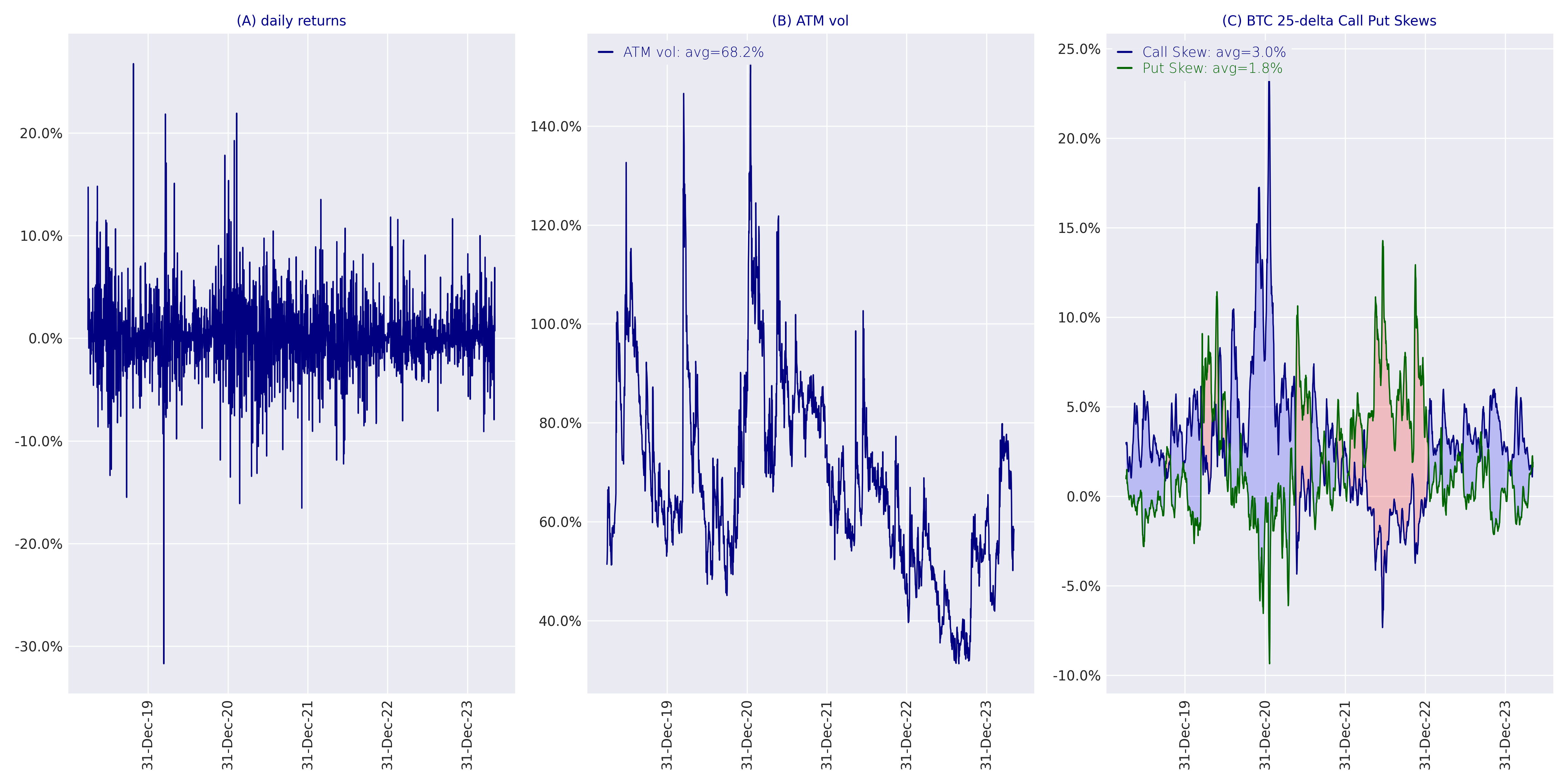}
    \label{fig:dynamics}
\end{figure}

To account for empirical features of crypto options markets, we develop a derivatives pricing model that incorporates two Hawkes processes, one for positive and one for negative jumps.
These processes capture the clustering of jumps, driven by self- and cross-excitation, effectively capturing the structure of jump arrival times. 
The model can be calibrated to the joint dynamics of BTC prices and of BTC implied volatility surfaces. 
In particular, we explore the resulting positive and negative jump risk premia, defined as the difference in expected jump magnitudes between the statistical measure and the risk-neutral measures.

Importantly, our model accommodates changing market preferences for skewness risk, when investors shift their demands between put options to hedge the downside risk and call options to capture positive upside. The preference for call options with upside potential is consistent with evidence from the equity market literature indicating some investors' preferences for assets offering lottery-like payoffs, see, e.g. \cite{bali2011maxing}, \cite{thaler1988anomalies}, \cite{kumar2009gambles}. In cryptocurrency option markets, such preferences for capturing the upside manifest in inflated prices of out-of-the-money call options during bullish periods, which produces positive risk premia of positive jumps.

Furthermore, the positive implied volatility skew is also present in the market data of so-called ``meme'' stocks. 
Options activity in these stocks is characterised by strong demand and high volumes for call options and by a positive spot-volatility correlation. 
\cite{gupta2025} introduce the ``Equity Euphoria Indicator'', which measures the percentage of stocks with a positive implied skew ranging from $0\%$ to $100\%$. 
The most recent values, 10\% as of January 2025 and 11.9\% as of September 2025, are close to the levels observed during the Dot.com bubble in late 1990 and the post-Covid rally in 2021. 
We notice that implied skews on ``meme'' stocks change sign to negative during severe bear market periods, so that structurally we cannot model the dynamics of these stocks using a stochastic volatility model with a constant correlation parameter. 
Adding jumps arriving at a constant intensity rate, as for the Bates model (\cite{bates1996jumps}), would not improve the structural consistency of the model dynamics, unless we specify stochastic intensity rates for positive and negative jumps. 
Thus, our model represents a simple construction to model the dynamics of stocks with varying signs of implied skew. 
It is straightforward to incorporate stochastic volatility in our framework; however, we skip this feature to reduce the number of model parameters.

Our work contributes to the existing literature in several areas. 
First, the model of the price process allows one to identify and understand the factors that drive variations in jump risk premia, in particular, how jump intensities adapt in response to the occurrence of jumps. 
Extending the work of \cite{hainaut2016bivariate,hainaut2017clustered} on univariate Hawkes process, we assume that the dynamics of these jump intensities are governed by a bivariate Hawkes process. 
Jump intensities increase with jump arrivals and decrease with time. 
As such, the model accounts for changes in asset prices and option prices stemming not only from movements in the underlying price, but also from variations in jump intensities.

Second, alongside specifying the pricing model, we derive the measure change between the statistical and risk-neutral measures. This allows us to derive the jump risk premia associated with positive and negative jumps as the difference between expected jump magnitudes under the risk-neutral and statistical measures. A jump risk premium therefore quantifies the difference between market expectations (from options prices) and market realisations (from the price history) and can be interpreted as a measure of market sentiment. As option prices (under the risk-neutral measure) reflect expectations about future jumps and their intensities, in our setting, investors pay a risk premium for call options due to positive jumps if they are optimistic (positive skewness of option prices). Likewise, investors pay a risk premium for put options due to negative jumps if they are pessimistic (negative skewness of option prices).

We further develop an estimation and calibration procedure that takes into account both historical BTC price dynamics and options data. In the time period considered (May 2019 to October 2023), both positive and negative risk premia are observed, with frequent periods where the implied jump risk is higher than the realised jump risk. The last observation is consistent with the findings on the variance risk premia, e.g. \citep{Carr2009,Granelli2016}.

Third, the cryptocurrency market is known to have at times a large basis or cost-of-carry in its futures prices, as well as large swings in the basis. 
For example, the annualised basis of one-month BTC futures on Deribit on 5 April 2021 was $35\%$, and -- after a sell-off on 19 May 2021 -- dropped to $-54\%$.
Both positive and negative carry are regularly observed in the market. 
These can only partially be explained by the funding costs associated with perpetual futures. 
Therefore, we investigate the relationship between the inferred jump risk premia and the carry of one-month BTC futures to determine whether risk preferences, beyond standard funding costs from perpetual contracts, significantly influence futures pricing dynamics.

Finally, the presence of jumps imposes challenges in risk management and hedging, both in terms of the magnitude of losses and in the unpredictability of jumps. Hedging strategies in crypto markets have been investigated by \citep[e.g.][]{liu2023hedging,matic2023hedging,lucicsepp2024}. Here, we investigate whether the risk premia, which capture market participants' forward-looking risk preferences, can be systematically earned through delta-hedged options trading strategies.

The Hawkes process, or self-exciting process, was introduced by \citep{Hawkes1971,Hawkes1974}, see also
\citep{Errais2010,Hawkes2018,Bacry2015,embrechts2011multivariate,hainaut2018hedging,hainaut2016bivariate} for reviews and applications of Hawkes processes in finance. 
There is a vast literature on jump risk and jump risk premia, \citep[e.g.][]{bates1996jumps,bates2000post,pan2002jump,bakshi2003delta,Eraker2004,broadie2007model,santa2010crashes,christoffersen2012dynamic,nguyen2019risk}. 
Skewness and their relation to risk premia are treated in \citep[e.g.][]{bakshi2003stock} with a focus on the options market, and from a behavioural perspective in \citep[e.g.]{bali2011maxing}. 
The literature on risk premia in cryptocurrency markets is growing \citep[e.g.][]{alexander2021bitcoin,alexander2023delta,almeida2024risk,dobrynskaya2024downside}.

The paper is organised as follows. 
 In Section \ref{sec:model}, we develop the model and derive a number of results, such as moment-generating functions and the equivalent measure change. 
 Section \ref{sec:two_stage} introduces the estimation method that accompanies the model. 
 In Section \ref{sec:jumps_premia}, we specify the jump risk premia. 
 Section \ref{sec:results} presents the empirical results, that is, exploring the historical jump risk premia, the relationship between the jump risk premia and the futures basis, as well as the ability of the jump risk premia to forecast one-day ahead option portfolio returns that can be used to improve hedging strategies. 
 All calculations in this work can be reproduced with codes available in Quantlet \href{https://github.com/QuantLet/HJP/tree/main/codes/}{\includegraphics[height=\baselineskip]{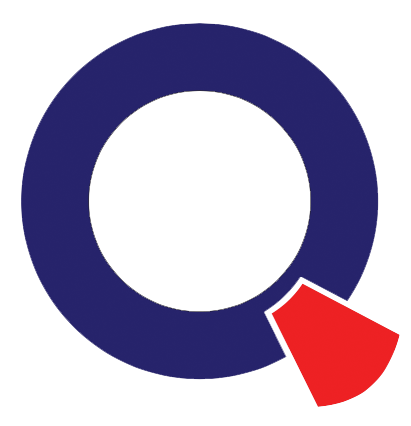}}.

\section{Model specification}\label{sec:model}

\subsection{Price dynamics under $\mathbb{P}$-measure}
Our model extends the univariate Hawks process dynamics considered by \cite{hainaut2018hedging} (also see \cite{hainaut2016bivariate}, \cite{hainaut2017clustered}, and chapter 4 of \cite{hainaut2022continuous} for variants of the univariate model adapted in different situations). To account for the empirical evidence of clustered positive and negative jumps, we incorporate two jump components with self-excitement and cross-excitement.

Let $\left(\Omega,\ \mathcal{F},\ \mathbb{P}\right)$ be a probability space with information filtration $\left(\mathcal{F}_t\right)$. The adapted price process $S_t$ is modeled as the sum of a continuous part driven by a Brownian motion and two jump components for positive and negative jumps, respectively. The intensities of the jump components are latent processes $\lambda^{+}_{t}$ and $\lambda^{-}_{t}$. 
The arrival times of positive and negative jumps are modeled by two counting processes, abbreviated by $N^{(1)}_{t}$ and $N^{(2)}_{t}$, respectively.

We introduce the following price dynamics under the statistical $\mathbb P$ measure:
\begin{align}
  \frac{\dd S_t}{S_{t-}} &= \mu \, \dd t  + \sigma \, \dd W_t
     \label{eq:dS}\\
    &\phantom{=\,}
    + \Big((\e^{J^+}-1)\, \dd N_t^{(1)}- \lambda_t^+ \E(\e^{J^+}-1) \, \dd t \Big)
    + \Big((\e^{J^-}-1)\, \dd N_t^{(2)}- \lambda_t^- \E(\e^{J^-}-1)\,
      \dd t\Big), \nonumber \\
  \begin{pmatrix}
    \dd \lambda_t^+\\
    \dd \lambda_t^-
  \end{pmatrix}
  &=
    \begin{pmatrix}
      \kappa^+ (\theta^+-\lambda_t^+)\\
      \kappa^- (\theta^--\lambda_t^-)
    \end{pmatrix}\, \dd t %
  +
  \begin{pmatrix}
    \beta_{11} & \beta_{12}\\
    \beta_{21} & \beta_{22}
  \end{pmatrix} %
                 \begin{pmatrix}
                   J^+\, \dd N_t^{(1)}\\
                   J^-\, \dd N_t^{(2)}
                 \end{pmatrix} \label{eq:bi-Hawkes},
\end{align}
where $\mu$ is a constant drift parameter, $\sigma$ is a constant volatility parameter, $W_t$ is a standard Brownian motion,  $N^{(1)}_t$ and $N^{(2)}_t$ are Poisson processes for the arrival of positive and negative jumps, respectively. Jump sizes $J^+$ and $J^-$ of positive and negative jumps with probability density functions (PDFs) $\varpi^+$ and $\varpi^-$, respectively, given by shifted exponential distributions $\varpi^+$ and $\varpi^-$. The jump parts in dynamics (\ref{eq:dS}) are adjusted by deterministic jump compensators, so the expected drift is determined by $\mu$. The coefficients $\beta_{ij}$, $i,j=1,2$, determine the level of excitement induced by jumps on the intensities.
Their signs are set by $\beta_{11}, \beta_{21}\geq0$ and $\beta_{12},
\beta_{22} \leq 0$, so that all jump realisations lead to jumps in intensity rates. We impose the following constraints on the jump parameters:
\begin{align}
\kappa^+ \geq \beta_{11}\E J^+ + \beta_{12}\E J^-, \quad
\kappa^- \geq \beta_{21}\E J^+ + \beta_{22}\E J^-.
\end{align}
These restrictions guarantee that the expected asymptotic values of the intensities are finite, 
 i.e., $\lim_{t\rightarrow \infty}\E\lambda^+_t<\infty$ and $\lim_{t\rightarrow \infty}\E\lambda^-_t<\infty$.
 For further details, we refer to Appendix \ref{sec:finite_expected_intensities}.

The solution to dynamics in (\ref{eq:dS}) is given by:
  \begin{align}
    S_t &= S_0 \, \exp\left(\left(\mu-\frac{\sigma^2}{2}\right) t + \sigma W_t \right)\\
    &\phantom{=\,} \cdot \exp\left(
     \Big(\sum_{j=1}^{N_t^{(1)}} J_j^- -\E(\e^{J^+}-1) \int_0^t \lambda_s^+\, \dd s\Big)
    + \Big(\sum_{j=1}^{N_t^{(2)}} J_j^- -\E(\e^{J^-}-1) \int_0^t \lambda_s^-\, \dd s\Big)
    \right),\nonumber\\
    \lambda_t^+ &= \theta^+ + \e^{-\kappa^+ t} (\lambda^+_0 - \theta^+)
                  + \beta_{11} \int_0^t \e^{-\kappa^+ (t-s)}\, J^+_s\,
                  \dd N_s^{(1)} %
                  + \beta_{12} \int_0^t \e^{-\kappa^+ (t-s)}\, J^-_s\,
                  \dd N_s^{(2)}\label{eq:lambda_p_sol},\\
    \lambda_t^- &= \theta^- + \e^{-\kappa^- t} (\lambda^-_0 - \theta^-)
                  + \beta_{21} \int_0^t \e^{-\kappa^- (t-s)}\, J^+_s\,
                  \dd N_s^{(1)} %
                  + \beta_{22} \int_0^t \e^{-\kappa^- (t-s)}\, J^-_s\,
                  \dd N_s^{(2)}\label{eq:lambda_m_sol}.
  \end{align}
We also specify the dynamics of the log return $X_t:=\ln \left(S_t/S_0\right)$ which is given by:
\begin{equation}\label{eq:dx}
\begin{split}
      \dd X_t  & = \left(\mu-\frac{\sigma^2}{2}\right) \dd t + \sigma \,\dd W_t \\
  & + \Big(J^+\, \dd N_t^{(1)} - \lambda_t^+ \E(\e^{J^+}-1) \, \dd t\Big) 
  + \Big(J^-\, \dd N_t^{(2)}- \lambda_t^-\E(\e^{J^-}-1)\, \dd t\Big).
  \end{split}
\end{equation}
In summary, the model incorporates three state variables $(X_t, \lambda^+_t, \lambda^-_t)$. The model parameters include a drift $\mu$ and the volatility $\sigma$ of the continuous component;
eight parameters for the intensities of positive and negative jumps $(\kappa^{+},\ \theta^{+},\ \beta_{11},\ \beta_{12})$ and $(\kappa^{-},\ \theta^{-},\ \beta_{21},\ \beta_{22})$, respectively;
four parameters parameters for PDFs for jumps in returns: $\varpi^{+}(j^{+})$ and $\varpi^{-}(j^{-})$, as defined next.

We assume that the sizes of positive and negative jumps are driven by shifted exponential distributions with the following PDFs:
\begin{align} \label{eq:jump}
\begin{split}
&\varpi^{+}(j)  = \frac{1}{\eta^{+}}\exp\left( -\frac{1}{\eta^{+}} (j-\nu^{+} )  \right) \mathbbm{1}_{\{j>\nu^{+}\}},\\
&\varpi^{-}(j)  = \frac{1}{\eta^{-}}\exp\left( \phantom{-}\frac{1}{\eta^{-}} (j-\nu^{-} )   \right) \mathbbm{1}_{\{j<\nu^{-}\}}, 
\end{split}
\end{align}
with $0<\eta^{+}$ and $0<\eta^{-}$ controlling the sizes of negative and positive jumps, respectively, and with $\nu^{+}>0$ and $\nu^{-}<0$ being thresholds.
The Laplace transforms of the jump size distributions (extended to $\R_-$ for negative jumps) are given by
  \begin{align}
    &&\mathcal L^{(+)}(\omega) &= \int_{\nu^+}^\infty \e^{-\omega j} \varpi^+(j)\, \dd j %
                      = \frac{\e^{-\nu^+ \omega}}{1 + \eta^+ \omega},&&
                      \quad\text{ if } \text{Re}(1/\eta^+ + \omega)>0 \label{eq:L_p}\\
    &&\mathcal L^{(-)}(\omega) &= \int_{-\infty}^{\nu^-} \e^{-\omega j} \varpi^-(j)\,
                      \dd j %
                      = \frac{\e^{-\nu^- \omega}}{1-\eta^-\omega}, &&\quad\text{
                      if } \text{Re}(\omega) < \text{Re}(1/\eta). \label{eq:L_m}
  \end{align}
where $\omega\in \mathbb{C}$ and $\text{Re}(\omega)$ denotes the real part.

Jump compensators are calculated using Equations (\ref{eq:L_p}) and (\ref{eq:L_m}) as follows: 
  \begin{align*}
    \E\left(\e^{J^+}-1\right) &= \int_{\nu^+}^\infty \left(\e^{j}-1\right) \varpi^+(j) \dd j = \frac{\e^{\nu^+}}{1-\eta^+}-1,\\
    \E\left(\e^{J^-}-1\right) &= \int^{\nu^-}_{-\infty} \left(\e^{j}-1\right) \varpi^-(j) \dd j = \frac{\e^{\nu^-}}{1+\eta^-}-1.
  \end{align*}

We specify threshold parameters $\nu^+$ and $\nu^-$ for flexibility in defining the minimum size of jumps for empirical inference. The thresholds $\hat \nu^+$ and $\hat \nu^-$ can be estimated using a peak-over-threshold (POT) approach discussed in Section \ref{sec:firststage}.

\subsection{Moment-generating function of joint dynamics}

The moment generating function (MGF) of the log-return variable $X_t=\ln (S_t/S_0)$ and jump intensities will be used for the equivalent measure change in Section \ref{sec:Q} and for the option pricing procedure in Section \ref{sec:second_stage}. The following Proposition extends Proposition 2.2. of \cite{hainaut2018hedging} to the bivariave Hawks process. 
\begin{proposition}\label{prop:MGF} 
    The MGF of joint variables $(X_T, \lambda^+_T, \lambda^-_T)$, $T\geq t$, generalised to complex arguments, is given by:
    \begin{equation}\label{eq:mgf_p}
    \begin{split}
        G^{\mathbb{P}}(\omega,\omega^+,\omega^-;T,x,\lambda^+, \lambda^-) & \equiv\E^{\mathbb{P}}\left(\exp(\omega X_T + \omega^+\lambda^+_T + \omega^-\lambda^-_T)\big|X_t, \lambda^+_t, \lambda^-_t\right)\\
        & = \exp\left(A(t,T)+ \omega X_t + C(t,T)\lambda_t^+ + D(t,T)\lambda_t^-\right), 
    \end{split}
    \end{equation}
    for $\omega, \omega^+, \omega^- \in \mathbb{C}$, where the functions $A, C,\text{and } D$ solve the following non-linear system of ordinary differential equations (ODEs):
    \begin{equation}
    \begin{split}
          \frac{\partial}{\partial t} A(t,T) &= - \mu \omega - \frac{\sigma^2}{2}\left(\omega^2-\omega\right) - \kappa^+\theta^+C(t,T) - \kappa^-\theta^-D(t,T), \\
          \frac{\partial}{\partial t} C(t,T) &= -\mathcal{L}^{(+)}(-\omega-C(t,T)\beta_{11}-D(t,T)\beta_{21}) + \kappa^+C(t,T) + \omega\Big(\frac{e^{\nu^+}}{1-\eta^+}-1\Big) + 1,\\
          \frac{\partial}{\partial t} D(t,T) &= -\mathcal{L}^{(-)}(-\omega-C(t,T)\beta_{12}-D(t,T)\beta_{22}) + \kappa^-D(t,T) + \omega\Big(\frac{e^{\nu^-}}{1+\eta^-}-1\Big) + 1,
    \end{split}
    \end{equation}
    where $A(T,T) =0$, $C(T,T)=\omega^+$, $D(T,T)=\omega^-$ with operators $\mathcal L^+$ and $\mathcal L^-$ defined in Equations \ref{eq:L_p} and \ref{eq:L_m}, respectively. 
\end{proposition} 
See Appendix \ref{sec:proofs} for the \hyperref[proof:prop:MGF]{proof}. 

\subsection{Equivalent risk-neutral measure $\mathbb{Q}$}\label{sec:Q}

We now construct the model dynamics under an equivalent martingale measure $\mathbb{Q}$ such that the discounted stock
price process is a martingale. This is an extension of \citep{hainaut2018hedging} to a bivariate self-exiting Hawks process.
Introduce the Radon-Nikodym derivative process defined by:
\begin{align}
  \frac{\text{d}\mathbb{Q}}{\text{d}\mathbb{P}}\Big|\mathcal{F}_t = \frac{M_t}{M_0},\quad t\geq 0,
\end{align}
where $\{M_t\}_{t \geq 0}$ is a local martingale under $\mathbb{P}$, specified below. By construction, we have the following:
\begin{equation}
\displaystyle \E^{\mathbb P}\left[\left(\frac{\dd \mathbb Q}{\dd \mathbb P}\big|\mathcal{F}_T\right)\right]=1, \ T \geq 0. 
\end{equation}
Define exponential processes of the form
\begin{equation}\label{eq:RD}
\begin{split}
  M_t\left(\xi^+, \xi^-, \varphi\right) \overset{\text{def}}= 
  &\exp\left\{
  \kappa^+_1(\xi^+)\lambda^+_t
  + \xi^+\sum_{j=1}^{N_t^{(1)}}J^+_j 
      + \kappa^-_1(\xi^-)\lambda^-_t 
  + \xi^-\sum_{j=1}^{N_t^{(2)}}J^-_j 
  \right\} \\
  &\cdot \exp \left\{
  q_2\left(\xi^+, \xi^-\right)t
  \right\}\\
  &\cdot \exp \left\{
  -\frac{1}{2}\int_0^t \varphi^2(s)ds - 
  \int_0^t \varphi(s)dW_s
  \right\}, 
  \end{split}
\end{equation}
where $\xi^+, \xi^- \in \mathbb{R}$ are two parameters that specify the risk-premia of positive and negative jumps, respectively, and that are to be estimated from options data. $q_1^+(\xi^+), q_1^-(\xi^-), q_2(\xi^+, \xi^-)$ are real-valued functions, and $\varphi(t)$ is an adapted process adapted to $\{\mathcal{F}_t\}_{t\geq 0}$. 
In Proposition \ref{prop:risk-neutralisation} below, we will choose 
\begin{equation}
\varphi(t)=\varphi + \varphi^+\lambda_t^+ + \varphi^-\lambda_t^-,
\end{equation}
where $\varphi, \varphi^+, \varphi^-\in \R$ are constants. 
We note that the process $M_t$ is strictly positive, so that, provided it is a martingale, the probability measures $\mathbb P$ and $\mathbb Q$ are absolutely continuous.
The following proposition specifies martingale conditions for $\left\{M_t\right\}_{t\geq 0}$.

\begin{proposition}\label{prop:martingaleConditions} 
 If, for any $\xi^+$ and $\xi^-$ $\in \mathbf{R}$, there exists a solution specified as functions $q_1^+(\xi^+)$,  $q_1^-(\xi^-)$ and $q_2(\xi^+, \xi^-)$ of the following system:
\begin{equation}\label{eq:M1}
\begin{split}
        q_1^+(\xi^+)\kappa^+\theta^+ +q_1^-(\xi^-)\kappa^-\theta^-+q_2(\xi^+, \xi^-) &=0,\\
        \mathcal L^{(+)}(-\chi^+)-1-q_1^+(\xi^+)\kappa^+ &=0,\\
        \mathcal L^{(-)} (-\chi^-)-1-q_1^-(\xi^-)\kappa^-&=0, 
  \end{split}
\end{equation}
then $M_t$ is a local martingale. Here, we introduce new variables $\chi^+$ and $\chi^-$ as follows: 
\begin{equation}\label{eq:chi}
\begin{split}
\chi^+ & = q_1^+(\xi^+)\beta_{11} + q_1^-(\xi^-)\beta_{21}+\xi^+,\\
\chi^- & = q_1^+(\xi^+)\beta_{12} + q_1^-(\xi^-)\beta_{22}+\xi^-.
  \end{split}
\end{equation}
\end{proposition}
See Appendix \ref{sec:proofs} for the \hyperref[proof:prop:martingaleConditions]{proof}.

We note that the system of equations in Eq.\ref{eq:M1} is non-linear because of its variables entering into non-linear functions $\mathcal L^{(+)}$ and $\mathcal L^{(-)}$. First, we make the following assumption that the risk-premium functions are linear.
\begin{assumption}
We assume that the real-valued functions are specified by means of a linear system of equations as follows:
\begin{equation}\label{eq:c}
\begin{split}
    q_1^+(\xi^+)      &= c^++\xi^+\\
    q_1^-(\xi^-)      &= c^-+\xi^-\\
    q_2(\xi^+, \xi^-) &= c+\xi^++\xi^-,
  \end{split}
\end{equation}
where $c^+, c^-, c \in \mathbb{R}$ are internal parameters which are determined so that the system in Eq (\ref{eq:M1}) is satisfied for any choice of $\xi^+$ and $\xi^-$. 
\end{assumption}

Second, we linearise the system of equations in \eqref{eq:M1}, by making the following assumption.

\begin{assumption}
The variables $\chi^+$ and $\chi^-$ in Eq. \eqref{eq:chi} are external variables that specify the risk premium, while variables $\xi^+$ and $\xi^-$ become internal parameters inferred using Eq.\eqref{eq:chi}.
\end{assumption}

\begin{corollary}
The specification of the risk-premium is defined by the two external parameters $\chi^+$ and $\chi^-$. The internal parameters of linear risk-premium functions in Eq. \eqref{eq:c}, internal risk premia $\xi^+$ and $\xi^-$ parameters in Eq. \eqref{eq:chi} along the martingale conditions in Eq. \eqref{eq:M1} are determined by five internal parameters $(\xi^+, \xi^-, c^+, c^-, c)$ which are obtained by solving the following linear system:
    \begin{equation}\label{eq:system}
        \begin{pmatrix}
        1+\beta_{11} &   \beta_{21} & \beta_{11} & \beta_{21} & 0\\
        \beta_{12} & 1+\beta_{22} & \beta_{12} & \beta_{22} & 0\\
        \kappa^+\theta^+ & \kappa^-\theta^- & \kappa^+\theta^+ & \kappa^-\theta^- & 1\\
        1&0&1&0&0\\
        0&1&0&1&0
        \end{pmatrix}  \begin{pmatrix}\xi^+\\ \xi^-\\ c^+\\ c^-\\ c \end{pmatrix} = 
        \begin{pmatrix}
        \chi^+ \\ \chi^- \\ 0 \\ \left(\mathcal{L}^{(+)}(-\chi^+)-1\right)\big/ \kappa^+\\ \left(\mathcal{L}^{(-)}(-\chi^-)-1\right)\big/ \kappa^-
        \end{pmatrix}.
    \end{equation}
\end{corollary}

\begin{proposition}\label{lemma:solution}
    The system of linear equations in Eq.\ \eqref{eq:system} in $(\xi^+, \xi^-, c^+, c^-, c)$ has a unique solution.
\end{proposition}
\begin{proof}
The determinant of the coefficient matrix in Eq.\eqref{eq:system} is equal to one, so it is invertible and a solution to the linear system exists and is unique.
\end{proof}

This chain of results establishes that the equivalent martingale measure $\mathbb{Q}$ is fully characterised by the two external risk-premium parameters $\chi^+$ and $\chi^-$, which can be interpreted as market-implied jump risk premia for positive and negative jumps, respectively. Given any such pair $(\chi^+, \chi^-)$, the corresponding internal parameters $(\xi^+, \xi^-, c^+, c^-, c)$ that define the Radon-Nikodym derivative are uniquely determined by solving the linear system in Equation~\eqref{eq:system}. The martingale property of the Radon-Nikodym derivative process $M_t$ is therefore always guaranteed under this construction.

Crucially, the transformation of the model from the physical measure $\mathbb{P}$ to the risk-neutral measure $\mathbb{Q}$ depends only on the values of $\chi^+$ and $\chi^-$. This leads to a clean and tractable separation between the statistical and risk-neutral dynamics, which then enables model calibration and risk-premium estimation from observed option prices. As shown in Proposition~\ref{prop:HawkesUnderQ}, the jump intensity dynamics under $\mathbb{Q}$ inherit a similar structure to those under $\mathbb{P}$, with appropriately adjusted parameters. Thus, the model maintains internal consistency and interpretability under both measures, while enabling flexible inference of jump risk premia from market data.

\begin{proposition}[Jump intensity dynamics under the risk-neutral measure $\q$ ]\label{prop:HawkesUnderQ}  
    Let 
    $N^{(1),\mathbb{Q}}_t$ and $N^{(2),\mathbb{Q}}_t$ be the counting processes of positive and negative jumps under the measure $\mathbb{Q}$ with intensities
    \begin{equation}
        \lambda^{+,\mathbb{Q}}_t=\mathcal{L}^{(+)}\left(-\chi^+\right) \lambda^+_t,\qquad
        \lambda^{-,\mathbb{Q}}_t=\mathcal{L}^{(-)}\left(-\chi^-\right) \lambda^-_t,
    \end{equation}
    where $\chi^+$ and $\chi^-$ are external risk-premium parameters in Eq.\eqref{eq:chi}.

The evolution of jump intensity rates  $\lambda^{+,\mathbb{Q}}_t$ and $\lambda^{-,\mathbb{Q}}_t$ under $\mathbb{Q}$ are driven by the following dynamics:
    \begin{equation}\label{eq:dxi}
\begin{split}
   \dd \lambda^{+,\mathbb{Q}}_t &= \kappa^+ \left(\theta^{+, \mathbb{Q}}-\lambda^{+,\mathbb{Q}}_t\right)\dd t
         + \beta^{\mathbb{Q}}_{11} J^{+, \mathbb{Q}}\, \dd N^{(1),\mathbb{Q}}_t
         + \beta^{-, \mathbb{Q}}_{12} J^{-, \mathbb{Q}}\, \dd N^{(2),\mathbb{Q}}_t\\
   \dd \lambda^{-,\mathbb{Q}}_t &= \kappa^- \left(\theta^{-, \mathbb{Q}}-\lambda^{-,\mathbb{Q}}_t\right)\dd t
         + \beta^{\mathbb{Q}}_{21} J^{+, \mathbb{Q}} \,\dd N^{(1),\mathbb{Q}}_t
         + \beta^{\mathbb{Q}}_{22} J^{-, \mathbb{Q}} \,\dd N^{(2),\mathbb{Q}}_t,
\end{split}
\end{equation}
    where 
    \begin{align}
         \theta^{+, \mathbb{Q}} &= \mathcal{L}^{(+)}(-\chi^+)\theta^+,\ 
         \beta^{\mathbb{Q}}_{11} = \mathcal{L}^{(+)}(-\chi^+)\beta_{11},\ 
         \beta^{\mathbb{Q}}_{12} = \mathcal{L}^{(+)}(-\chi^+)\beta_{12},\\
         \theta^{-, \mathbb{Q}} &= \mathcal{L}^{(-)}(-\chi^-)\theta^-,\ 
         \beta^{\mathbb{Q}}_{21} = \mathcal{L}^{(-)}(-\chi^-)\beta_{21},\ 
         \beta^{\mathbb{Q}}_{22} = \mathcal{L}^{(-)}(-\chi^-)\beta_{22},
    \end{align}
    Here, $J^{+, \mathbb{Q}}$ and  $J^{-, \mathbb{Q}}$ are realizations of positive and negative jumps from distributions with probability density functions(PDFs)
    $\varpi^{+, \mathbb{Q}}(j^+)$ and $\varpi^{-, \mathbb{Q}}(j^-)$:
    \begin{align*}
        \varpi^{+,\mathbb{Q}}(j) &= \frac{1}{\eta^{+, \mathbb{Q}}}\exp\left(-\frac{1}{\eta^{+, \mathbb{Q}}}(j-\nu^+)\right)\mathbbm{1}_{\{j > \nu^+\}},\ 
        \varpi^{-,\mathbb{Q}}(j) = \frac{1}{\eta^{-, \mathbb{Q}}}\exp\left(\phantom{-}\frac{1}{\eta^{-, \mathbb{Q}}}(j-\nu^-)\right)\mathbbm{1}_{\{j < \nu^-\}}
    \end{align*}
    where 
    \begin{align}
        \eta^{+,\mathbb{Q}} 
        &= \frac{\eta^+}{1-\eta^+\chi^+},\ 
        \eta^{-,\mathbb{Q}} 
        = \frac{\eta^-}{1+\eta^-\chi^-}.
    \end{align}
\end{proposition} 
See Appendix \ref{sec:proofs} for the \hyperref[proof:HawkesUnderQ]{proof}.
The proposition \ref{prop:HawkesUnderQ} indicates that the dynamics of the intensities and the PDFs of the jump sizes are specified by the two external parameters $\chi^+$ and $\chi^-$. Accordingly, for model calibration to implied volatility data, we need to infer a pair of value of $\chi^+$ and $\chi^-$ that provides the best fit. The following proposition specifies $\varphi(t)$ such that an equivalent risk-neutral measure is obtained. 
\begin{proposition}\label{prop:risk-neutralisation}
    We introduce the adapted process $\varphi(t)$ as follows:
    \begin{align}
        \varphi(t)      
        &= \frac{\mu - r}{\sigma} %
        - \frac{
         \E(e^{J^+}-1) 
        - \mathcal{L}^{(+)}(-\chi^+)  \E^\mathbb{Q}(e^{J^{+,\q}}-1)}
        {\sigma}\, 
        \lambda_t^+ %
        - \frac{
         \E(e^{J^-}-1)
        +\mathcal{L}^{(-)}(-\chi^-) \E^\mathbb{Q}(e^{J^{-,\q}}-1)
        } {\sigma }\lambda_t^-,\nonumber
        \\[5pt]
        &= \frac{\mu - r}{\sigma} %
        - \frac{
        \lambda^{+}_t \E(e^{J^+}-1)
        - \lambda^{+,\mathbb{Q}}_t \E^\mathbb{Q}(e^{J^{+,\q}}-1)
        } {\sigma} %
        - \frac{
        \lambda^{-}_t \E(e^{J^-}-1)
        - \lambda^{-,\mathbb{Q}}_t \E^\mathbb{Q}(e^{J^{-,\q}}-1)
        } {\sigma }\label{eq:varphi2},
    \end{align}
    where $r$ is the risk free rate. Then the log return process $X_t=\ln (S_t/S_0)$ is driven under $\mathbb{Q}$ by:

\begin{equation}\label{eq:dxq}
\begin{split}
      \dd X_t  & = \left(r-\frac{\sigma^2}{2}\right) \dd t + \sigma \,\dd W^{\mathbb{Q}}_t \\
      & + \Big(J^{+, \mathbb{Q}}\, \dd N_t^{\mathbb{Q}, (1)} - \lambda_t^{+, \mathbb{Q}} \E^\mathbb{Q}(\e^{J^{+, \mathbb{Q}}}-1) \, \dd t\Big) 
  + \Big(J^{-, \mathbb{Q}}\, \dd N_t^{\mathbb{Q}, (2)}- \lambda_t^{-, \mathbb{Q}}\E^\mathbb{Q}(\e^{J^{-, \mathbb{Q}}}-1)\, \dd t\Big),\\
  \end{split}
    \end{equation}
    where $W^{\mathbb{Q}}$ is a standard Brownian motion under $\mathbb{Q}$, and the discounted price process is a $\mathbb Q$-martingale.
\end{proposition}
See Appendix \ref{sec:proofs} for the \hyperref[proof:prop:risk-neutralisation]{proof}.
The function $\varphi(t)$ specifies the risk premia associated with the continuous component of the process. We will discuss the jump risk premia in Section \ref{sec:jumps_premia}.

\begin{corollary}[MGF in Proposition \ref{prop:MGF} under $\mathbb{Q}$] \label{prop:MGF_q} 
    The MGF of log-return $X_T$ with the dynamics in Eq.\eqref{eq:dxq} under $\mathbb{Q}$ is given by:
    \begin{equation}\label{eq:mgf_q}
    \begin{split}
        G^{\mathbb{Q}}(\omega;T,x,\lambda^{+, \mathbb{Q}}, \lambda^{-, \mathbb{Q}}) & \equiv\E^{\mathbb{Q}}\left(\exp(\omega X_T)\big|X_t, \lambda^{+, \mathbb{Q}}_t, \lambda^{-, \mathbb{Q}}_t\right)\\
        & = \exp\left(A(t,T)+ \omega X_t + C(t,T)\lambda_t^{+, \mathbb{Q}} + D(t,T)\lambda_t^{-, \mathbb{Q}}\right)\\
        & = \exp\left(\omega X_t \right)E(t, T; \omega) , 
    \end{split}
    \end{equation}
    where functions $A, C,\text{and } D$ solve Eq.\eqref{eq:mgf_p} with $\omega^+=0, \ \omega^-=0$ and with the following substitution of model parameters:
\begin{equation}\label{eq:mgf_q1}
\begin{split}
& \mu \rightarrow r,\\
& \theta^{+} \rightarrow  \theta^{+, \mathbb{Q}} = \mathcal{L}^{(+)}(-\chi^+)\theta^+, \ \beta_{11} \rightarrow \beta^{\mathbb{Q}}_{11} = \mathcal{L}^{(+)}(-\chi^+)\beta_{11},\ \beta_{12}\rightarrow \beta^{\mathbb{Q}}_{12} = \mathcal{L}^{(+)}(-\chi^+)\beta_{12},\\
& \theta^{-} \rightarrow \theta^{-,\mathbb{Q}} = \mathcal{L}^{(-)}(-\chi^-)\theta^-, \  \beta_{21}\rightarrow \beta^{\mathbb{Q}}_{21} = \mathcal{L}^{(-)}(-\chi^-)\beta_{21}, \ \beta^{\mathbb{Q}}_{22} \rightarrow \beta^{\mathbb{Q}}_{22} = \mathcal{L}^{(-)}(-\chi^-)\beta_{22},\\
& \eta^{+} \rightarrow  \eta^{+,\mathbb{Q}} =  \frac{\eta^+}{1-\eta^+\chi^+},\ \eta^{-}\rightarrow \eta^{-,\mathbb{Q}}   = \frac{\eta^-}{1+\eta^-\chi^-}.
\end{split}
\end{equation}
\end{corollary}

\begin{corollary}\label{prop:option_pricing} Using the Lewis-Lipton formula (\cite{Lewis2001}, \cite{Lipton2002}), the value of a call option, denoted by $U^{call}$, and a put option, denoted by $U^{put}$, with strike price $K$ and maturity time $T$ is computed for log-price dynamics in Eq.\eqref{eq:dxq} under the risk-neutral measure $\mathbb{Q}$ by:
\begin{equation} \label{eq:opt1}
\begin{split}
  & U^{call}(\tau, S,K)  = e^{-r\tau}S- U(\tau, X, K), \\
	& U^{put}(\tau, S,K)  = e^{-r\tau} K - U(\tau, X, K),
\end{split} 
\end{equation}
where $\tau=T-t$. Here, $U(\tau, X, K)$ is the price of the capped payoff $\min(S,K)$ which is computed by:
\begin{equation} \label{eq:opt2}
\begin{split}
  U(\tau, X, K)&  = \frac{e^{-r\tau} K}{\pi} \Re\left[\int^{\infty}_{0} e^{-(iy-1/2) X^{*}}  \frac{1}{y^{2}+1/4} E(\tau; \omega=iy-1/2)dy\right],
\end{split} 
\end{equation}
where $X^{*}=\ln(S/K)$ is log-moneyness, $\Re(x)$ is the real part of $x$, and $E(\tau;\omega)$ is specified in Eq.\eqref{eq:mgf_q}.
\end{corollary} 
For the proof, see the proof of Proposition 5.1. in \cite{sepprakhmonov2023}. 

For numerical implementation, the system of ODEs in Eq.\eqref{eq:mgf_p} must be solved numerically using standard Runge-Kutta methods. Given that the computation of the MGF part $E(\tau;\omega)$ does not depend on the strike price of the option, the function $E(\tau;\omega)$ can be computed on a grid and apply for the numerical calculation of call and put options at the same maturity time with different strike prices\footnote{Github project StochVolModels provides Python code with prototype implementation of option valuation using Eqs \eqref{eq:opt1} and \eqref{eq:opt2}, see \url{https://github.com/ArturSepp/StochVolModels/blob/main/stochvolmodels/pricers/hawkes_jd_pricer.py}.}.

\section{Two stage approach calibration of the model}\label{sec:two_stage}
We adopt the following two-stage calibration approach. First, we estimate the model parameters under statistical measure $\mathbb{P}$ from historical price data. Then, we infer the risk premia and the model parameters under the risk-neutral measure $\mathbb{Q}$ using the options data.

\subsection{Time series estimation}\label{sec:firststage}
To filter out jumps from time series of returns, we implement a POT procedure similarly to the approach in \cite{embrechts2011multivariate} and \cite[Chapter 4]{hainaut2022continuous}. We consider a sample of equally-spaced log-returns observed in the time window $[0,\ T]$.
 The POT procedure labels the log-returns that exceed the threshold as jumps. The underlying assumption is that the continuous part of the price process has normally distributed log-returns which implies that sampled returns, filtered by excluding jumps, have zero skewness and zero excess kurtosis. 

Let $(X^{(\nu^+,\nu^-)}):=\left\{X_t: \nu^- < X_t < \nu^+\right\}$ the set of log-returns bounded by $\nu^-$ and $\nu^+$. 
 The estimate of the thresholds is then given by:
 \begin{align}
  \left(\widehat\nu^+,\ \widehat\nu^-\right) = \underset{\nu^+, \nu^-}{\text{argmin}}  \left[
    |\text{skew}(X^{(\nu^+,\nu^-)})|+|\text{kurt}(X^{(\nu^+,\nu^-)})|
    \right],  \end{align}
 where $\text{skew}(X) = |X|^{-1}\sum_{x \in X}(x-\hat \mu_X)^3/\hat \sigma_X^3$, 
       $\text{kurt}(X) = |X|^{-1}\sum_{x \in X}(x-\hat \mu_X)^4/\hat \sigma_X^4 - 3$, 
      $\hat \mu_X$ and $\hat \sigma_X$ are the sample mean and standard deviation, respectively, and $|X|^{-1}$ is the sample size.
  Given estimated thresholds, we identify the set of positive jump events $\mathcal{J}^+$ and the set of negative jump events $\mathcal{J}^-$ respectively as follows:
\begin{equation}
    \begin{split}
    \mathcal{J}^+ &= \left\{J^+(s)\right\}_{s \leq T} = \left\{
      X_s \in X: X_s \geq \widehat\nu^+ \text{ and } s \leq T
      \right\},\\
      \mathcal{J}^- &= \left\{J^-(s)\right\}_{s \leq T} = \left\{
        X_s \in X: X_s \leq \widehat\nu^- \text{ and } s \leq T
        \right\}.
    \end{split}
\end{equation}

We construct the counting processes for positive jumps and for the negative jumps as follows:
\begin{equation}
    \begin{split}
    \widehat{N^{(1)}} &= \#\left\{
        X_s \in X: X_s \geq \widehat\nu^+ \text{ and } s \leq T
        \right\}, \\ 
    \widehat{N^{(2)}} &= \#\left\{
        X_s \in X: X_s \leq \widehat\nu^- \text{ and } s \leq T
        \right\}.
    \end{split}
\end{equation}
  
We denote the corresponding sequences of arrival times of positive jumps $\mathcal{T}^+$, of negative jumps $\mathcal{T}^-$, and an ordered union of jump arrival times $\mathcal{T}^{\pm}$ by:
\begin{equation}
    \begin{split}
& \mathcal{T}^+ = \left\{t \in [0,\ T]: X_t \geq \hat \nu^+ \right\} = \left\{T^+_1,\ T^+_2,\ ...,\ T^+_{\widehat{N^{(1)}}}\right\},\\
& \mathcal{T}^- = \left\{t \in [0,\ T]: X_t \leq \hat \nu^- \right\} = \left\{T^-_1,\ T^-_2,\ ...,\ T^-_{\widehat{N^{(2)}}}\right\},\\
& \mathcal{T}^{\pm} = \left\{T^{\pm}_{[1]},\ T^{\pm}_{[2]},\ ...,\ T^{\pm}_{[\widehat{N^{(1)}}+\widehat{N^{(2)}}]}\right\}.
    \end{split}
\end{equation}

Jumps sizes are independent of other random variables and follow shifted exponential distributions. 
Therefore, given the shift parameter estimates $\widehat \nu^+$ and $\widehat \nu^-$ from the POT procedure, the estimators of $\eta^+$ and $\eta^-$ are given by averaging:
\begin{equation}\label{eq:mean_p_jumps_sizes}
    \begin{split}
\widehat \eta^+ &= \widehat{N^{(1)}}^{-1} \sum_{j \in \mathcal{J}^+}\left(j - \widehat \nu^+\right),\\
\widehat \eta^- &= -\widehat{N^{(2)}}^{-1} \sum_{j \in \mathcal{J}^-}\left(j + \widehat \nu^-\right). 
    \end{split}
\end{equation}

Next, we estimate the parameters of the intensities dynamics via the maximum likelihood estimator (MLE). We deploy a version of the multivariate log-likelihood function documented in \cite{embrechts2011multivariate} as follows:
 \begin{align}
  \ln L' &= \sum_{T^+ \in \mathcal{T}^+} \ln \lambda^+(T^+-)\varpi^+(J^+(T^+))
          + \sum_{T^- \in \mathcal{T}^-} \ln \lambda^-(T^--)\varpi^-(J^-(T^-)) \nonumber\\
  &- \int_0^T \lambda^+(t-)\text{d}t 
  - \int_0^T \lambda^-(t-)\text{d}t. 
 \end{align}
We note that the likelihood takes the left-continuous version (indicated by $T-$) of the intensities processes (see \citep[p. 232]{daley2003introduction}). Since distributions of jumps sizes are already estimated throughout the POT procedure using Eq.(\ref{eq:mean_p_jumps_sizes}), we only need a partial likelihood for the intensity processes defined by: 
\begin{align}
  \ln L = \sum_{T^+ \in \mathcal{T}^+} \ln \lambda^+(T^+-)
  + 
   \sum_{T^- \in \mathcal{T}^-} \ln \lambda^-(T^--)
  - \int_0^T \lambda^+(t-)\text{d}t
  - \int_0^T \lambda^-(t-)\text{d}t. \label{eq:likelihood}
 \end{align}
The relationship between the intensities and the model parameters is specified for $s \in \left[T_{[k-1]}^\pm,\ T_{[k]}^\pm\right)$ and $k \in \left\{1,2,...,\ \widehat{N^{(1)}}+\widehat{N^{(2)}}\right\}$ as follows:
 \begin{align}
  \lambda^+(s) &= \theta^+ + e^{-\kappa^+\left(s - T_{[k-1]}^\pm\right)}\left(\lambda^+(T_{[k-1]}^\pm)-\theta^+\right),\\
  \lambda^-(s) &= \theta^- + e^{-\kappa^-\left(s - T_{[k-1]}^\pm\right)}\left(\lambda^-(T_{[k-1]}^\pm)-\theta^-\right).
 \end{align}
 
In the event of positive jump $T^+ \in \mathcal{T}^+$, the intensities jump by:
\begin{equation}
    \begin{split}
  \lambda^+(T^+) &= \lambda^+(T^+-) + \beta_{11}J^+(T^+),\\
  \lambda^-(T^+) &= \lambda^-(T^+-) + \beta_{21}J^+(T^+).
    \end{split}
\end{equation}

In the event of a negative jump $T^- \in \mathcal{T}^-$, the intensities jump by:
\begin{equation}
    \begin{split}
  \lambda^+(T^-) &= \lambda^+(T^--) + \beta_{12}J^-(T^-),\\
  \lambda^-(T^-) &= \lambda^-(T^--) + \beta_{22}J^-(T^-).
    \end{split}
\end{equation}

The integrals of the intensities in the partial likelihood Eq.(\ref{eq:likelihood}) can be computed as follows:
\begin{equation}
    \begin{split}
  \int_0^T\lambda^+(t-)\text{d}t &= \sum_{k=1}^{\widehat{N^{(1)}} + \widehat{N^{(2)}}} \int_{T^{\pm}_{[k-1]}}^{T^\pm_{[k]}}\lambda^+(t-)\text{d}t\\
  &= \sum_{k=1}^{\widehat{N^{(1)}} + \widehat{N^{(2)}}} 
 \theta^+(T_{[k]}^\pm-T_{[k-1]}^\pm)
 +\left(\lambda^+(T_{[k]}^\pm-)-\theta^+\right)\frac{1-e^{-\kappa^+(T_{[k]}^\pm - T_{[k-1]}^\pm)}}{\kappa^+},\\ 
 \int_0^T\lambda^-(t-)\text{d}t&= \sum_{k=1}^{\widehat{N^{(1)}} + \widehat{N^{(2)}}} 
\theta^-(T_{[k]}^\pm-T_{[k-1]}^\pm)
+\left(\lambda^-(T_{[k]}^\pm-)-\theta^-\right)\frac{1-e^{-\kappa^-(T_{[k]}^\pm - T_{[k-1]}^\pm)}}{\kappa^-}.
    \end{split}
\end{equation}

Therefore, given a set of model parameters, the intensities and their integrals can be computed quickly in a recursive way starting from time 0. We refer to Section 5.2 of \cite{laub2021elements} and references therein for the method of directly computing the likelihood.
  We apply a numerical optimiser to obtain estimates of the jump parameters, $\hat \kappa^\pm, \ \hat\theta^\pm, \{\hat \beta_{ij}\}_{i,j=1,2}$, that maximise the likelihood. 
Figure \ref{fig:empirical?intensities} illustrates the estimation procedure applied to BTC daily returns from 2019-05-30 to 2023-10-03.

\begin{figure}
    \centering
    \caption{Subplot (A) shows BTC daily returns which are classified using the POT procedure as positive jumps (green), negative jumps (red), or non-jumps (yellow). Subplot (B) displays the model intensities estimated via MLE, calibrated using the interarrival times and magnitudes of the jumps identified by POT.}
    \includegraphics[width=1\linewidth]{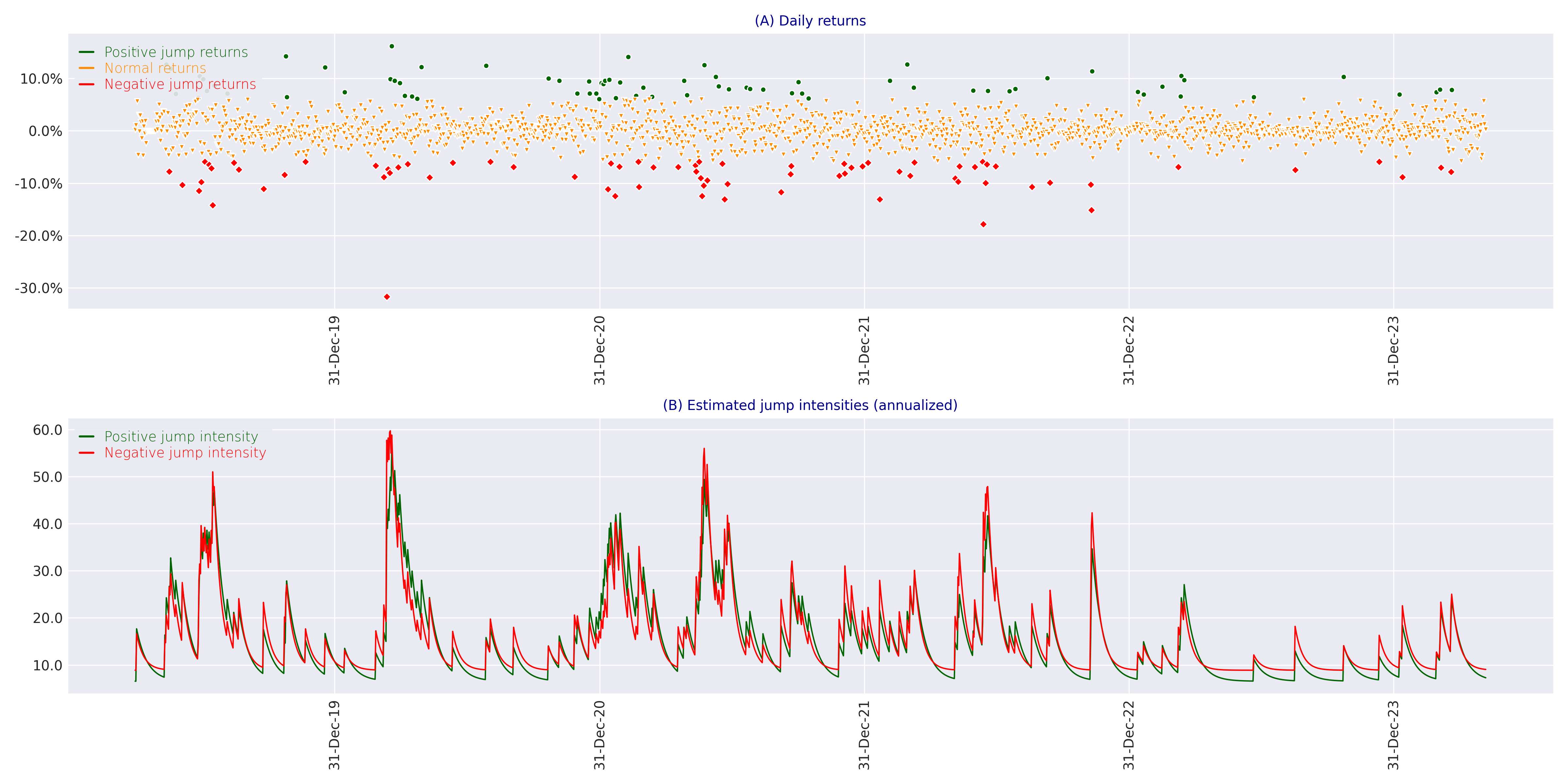}
    \label{fig:empirical?intensities}
\end{figure}

\subsection{Calibration to market implied volatilities}\label{sec:second_stage}
Given a set of European option prices and corresponding set of implied volatilities, we deduce volatility $\sigma$, and risk-premia $\chi^+$ and $\chi^-$ by minimising the mean absolute percentage error (MAPE) as follows:
 \begin{align*}
  \left(\widehat \sigma,\ \widehat \chi^+,\ \widehat\chi^-\right) = \underset{\sigma,\chi^+,\chi^-} {\text{argmin}}
  \sum_{i=1}^I
  w_i \frac{\left|\text{IV}_{\text{market}}\left(\tau_i,\ K_i,\ O_i\right)-\text{IV}_\text{model}\left(\tau_i,\ K_i,\ O_i;\ \sigma,\ \chi^+,\ \chi^-\right)\right|}{\text{IV}_{\text{market}}\left(\tau_i,\ K_i,\ O_i\right)},
 \end{align*}
 where $I$ is the sample size, $w_i$ is the weight applied to the $i^\text{th}$ option,
 $\text{IV}_\text{market}\left(\tau_i,\ K_i,\ O_i\right)$ is the Black-Scholes (BS) implied volatility of the $i^\text{th}$ option with time-to-maturity $\tau_i$, strike $K_i$, and option type $O_i\in \left\{\text{Call},\ \text{Put}\right\}$. Here,
 $\text{IV}_\text{model}\left(\tau_i,\ K_i,\ O_i;\ \sigma,\ \chi^+,\ \chi^-\right)$ is the corresponding BS implied volatility produced by our model given $\sigma$, $\chi^+$, and $\chi^-$. 
 We take the market BS vega of each option as weight $w$, thus assigning greater importance to options whose prices are more sensitive to mispricing of implied volatilities.

The computation of model implied volatilities given $\sigma$, $\chi^+$ and $\chi^-$ consists of the following steps.

\begin{enumerate}
    \item We compute $\lambda^{\pm,\mathbb{Q}}_t$, $\kappa^{\pm, \mathbb{Q}}$, $\theta^{\pm, \mathbb{Q}}$, $\beta_{ij}^\mathbb{Q}$ for $i,j\in [1,2]$, and $\eta^{\pm, \mathbb{Q}}$ using transformation in Eq (\ref{eq:mgf_q1}).

    \item We apply the MGF under $\mathbb{Q}$ in Eq.\ref{eq:mgf_q} along with risk-neutral drift $r$ and volatility $\sigma$ to value call and put options using Eq.\ref{eq:opt1}.
    \item We infer model implied volatilities by inverting the Black-Scholes formula. 
\end{enumerate}

We use a numerical optimiser to estimate $\widehat \sigma,\ \widehat \chi^+,\ \widehat\chi^-$ that minimise the MAPE. 
 
\section{Positive and negative jump risk premia}\label{sec:jumps_premia}

Following the literature on variance risk premia \citep[e.g.][]{Carr2009,bollerslev2009expected,almeida2024risk}, we define jump risk premia as the difference between expected jump sizes under the statistical measure and expected jump sizes under the risk-neutral measure.

\begin{corollary}[Using Proposition \ref{prop:risk-neutralisation}]
The positive jump risk premium denoted by $\gamma^+_t$ and negative jump risk premium denoted $\gamma^-_t$ are specified by: 
\begin{align}
    \gamma^+_t & = \lambda^{+, \mathbb{Q}}_t \E^\mathbb{Q}\big(e^{J^{+,\q}}-1\big) -  \lambda^{+}_t \E\big(e^{J^+}-1\big),\label{eq:gamma_p}\\
    \gamma^-_t &= \lambda^{-, \mathbb{Q}}_t \E^\mathbb{Q}\big(e^{J^{-,\q}}-1\big) -  \lambda^{-}_t \E\big(e^{J^-}-1\big)\label{eq:gamma_m}.
\end{align}
\end{corollary}

According to the model specification for the price process in Equation~(\ref{eq:dS}), 
the terms $\lambda_t^+\E(e^{J^+}-1)$ and $\lambda_t^-\E(e^{J^-}-1)$ are the compensators for positive and negative jumps under statistical measure $\p$, respectively. By Proposition \ref{prop:HawkesUnderQ}, the parameters $\chi^+$ and $\chi^-$ in Eq.\eqref{eq:chi} specify compensators $\lambda^{+, \mathbb{Q}}_t \E^\mathbb{Q}\big(e^{J^{+,\q}}-1\big)$ and $\lambda^{-, \mathbb{Q}}_t \E^\mathbb{Q}\big(e^{J^{-,\q}}-1\big)$ under the risk-neutral measure $\q$. Hereby, the parameters $\chi^+$ and $\chi^-$ are inferred from options data.

In this construction, the jump risk premia reflect the difference in forward-looking risk preferences inferred from option prices compared to risk preferences inferred from historical price dynamics. Positive $\gamma_t^+$ and negative $\gamma_t^-$ indicate that the market expects higher future jump risk than was observed in the past, and vice versa for negative $\gamma_t^+$ and positive $\gamma_t^-$.
The jump risk premia are expressed in the same units as the drift rate $\mu$, i.e., as an annualised yield rate.

The separate treatment of positive and negative jumps is motivated by the fact that the market reacts differently to up and down moves in the price process. In particular, the left tail of the risk-neutral distribution is known to be more vulnerable to market downturns than the right tail to market upswings.  This is also seen in the impulse response functions in Section \ref{sec:jumps_premia_evolution} below, where the responses of positive and negative jump risk premia differ in magnitude and speed of decay in impact.

Figure \ref{fig:IV_smiles_varying_gammas} depicts the impact of jump risk premia on the implied volatility profiles of model-generated options with one month to maturity.
Generally, the magnitude of the jump risk premia is directly correlated with the overall level of model implied volatilities.
A higher positive jump risk premium and a lower negative jump risk premium contribute to a higher level of model implied volatilities. 
However, as the panel on the top left reveals, the influence of
$\gamma^+$ is more pronounced on options with a moneyness greater than
1, predominantly in-the-money (ITM) call options.
The panel on the right shows the effect of the negative jump risk premium, where one has to bear in mind that a negative $\gamma^-$ indicates that more negative jump activity is expected under $\q$ than under $\p$. In turn, the more negative the $\gamma^-$ is, the higher the level of the model implied volatilities. 
Similarly, this change in $\gamma^-$ shows a stronger impact on options with a moneyness less than 1, typically where the ITM put options are concentrated.

Figure \ref{fig:RND_variation_under_diff_jump_premia} illustrates the impact of $\gamma^{+}$ and $\gamma^{-}$ on the risk-neutral probability density of one-month log returns. The figure reveals that $\gamma^{+}$ and $\gamma^{-}$ jointly govern the asymmetry and tail thickness of the risk-neutral density, encapsulating the market’s pricing of upward and downward jump risks.

\begin{figure}[!t]
    \caption{
        The panels show model implied volatilities of one-month-maturity options under different positive and negative jump risk premium levels while keeping the opposite jump risk premium zero.
        The levels of positive jump risk premium are indicated to the right of the corresponding model implied volatilities. The model implied volatilities in red are under zero positive and negative jump risk premia. \href{https://github.com/QuantLet/HJP/tree/main/codes/jupyterNotebooks/IV_smiles_under_diff_risk_premia_level}{\includegraphics[height=\baselineskip]{qletlogo_tr.png}IV\_smiles\_under\_diff\_risk\_premia\_level}
    }
    \includegraphics[width=\textwidth]{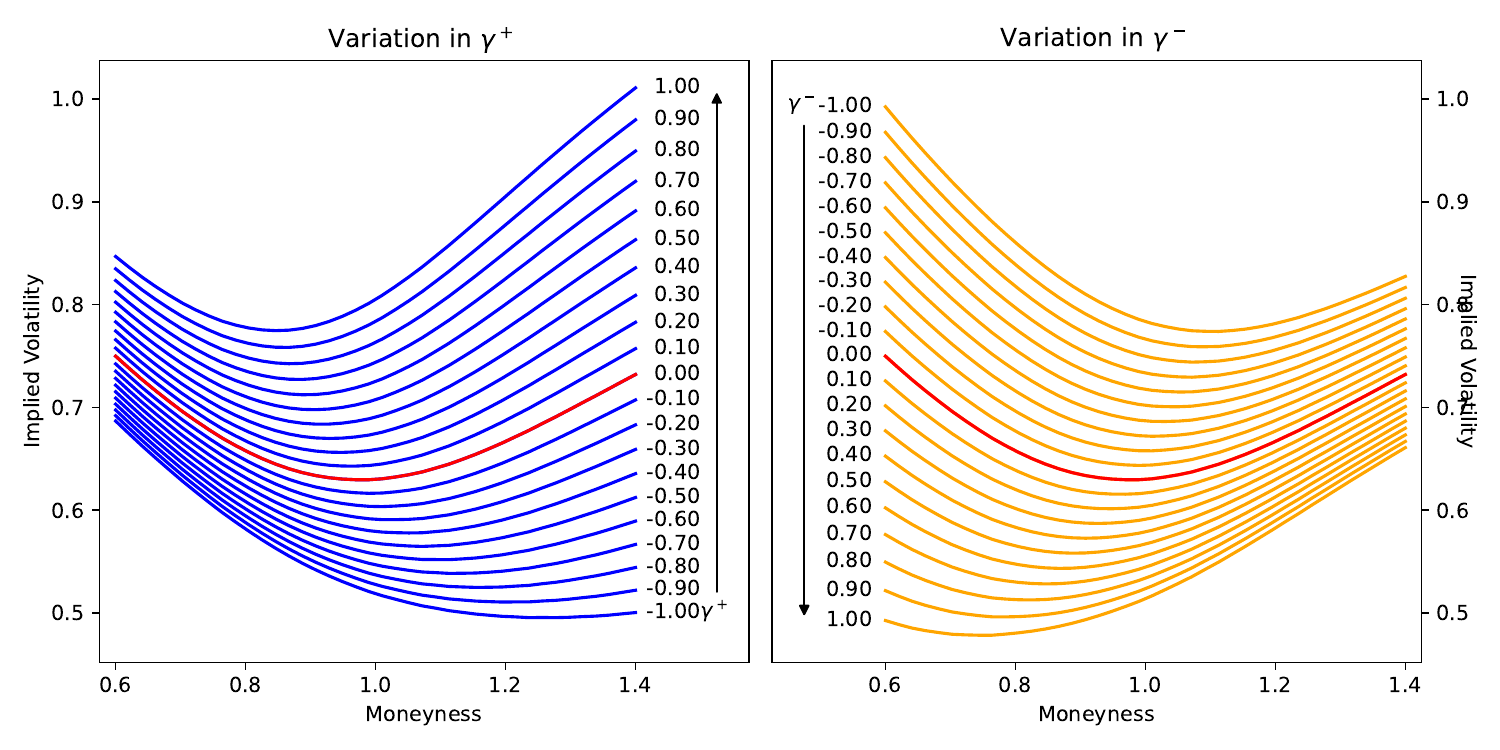}
    \label{fig:IV_smiles_varying_gammas}
\end{figure}

\begin{figure}[!t]
    \caption{
        The panels display the risk-neutral probability densities of one-month log return under varying levels of the positive and negative jump risk premia, $\gamma^{+}$ and $\gamma^{-}$. The left panel illustrates the effect of changes in risk-premia for calls $\gamma^{+}$ while holding $\gamma^{-}=0$, and the right panel shows the effect of changes in risk-premia for puts $\gamma^{-}$ while holding $\gamma^{+}=0$. Increasing $\gamma^{+}$ fattens the right tail of the distribution, leading to greater right-skewness in the distribution. In contrast, variation in $\gamma^{-}$ produces the opposite pattern: a fatter or thinner left tail depending on the sign of the premium. Together, $\gamma^{+}$ and $\gamma^{-}$ govern the asymmetry and tail behavior of the risk-neutral density, reflecting how markets price the risk of upward and downward jumps.
    }
    \includegraphics[width=\textwidth]{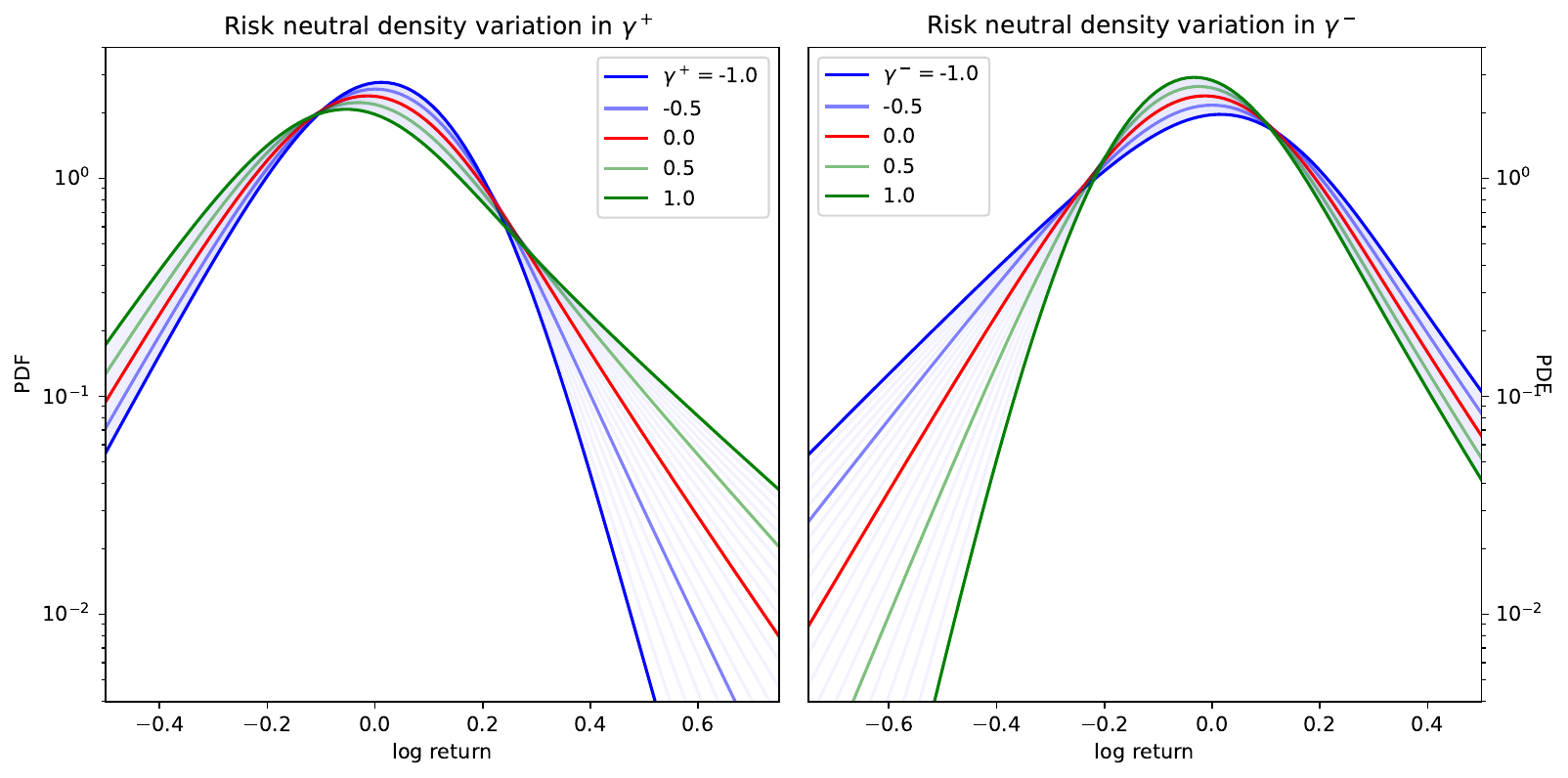}
    \label{fig:RND_variation_under_diff_jump_premia}
\end{figure}

\section{Empirical Results}\label{sec:results}
\subsection{Data and empirical study design}
We study the options market for cryptocurrencies, with its dynamics heavily influenced by sentiment-driven jumps and severe price adjustments. In particular, we consider options on Bitcoin (BTC) which are the most liquid in cryptocurrency derivatives. We gather data on both BTC price and BTC options from Deribit, a crypto derivatives exchange established in June 2016. The exchange offers options and futures trading for both BTC and ETH and provides extensive data on BTC prices and option prices through its API. Deribit's comparatively low fees for market makers contribute to the open interest and transaction volume for options traded on the platform. 

Our sampled BTC price data spans from 2015-12-31 to 2023-10-04 with the sample of BTC opening prices at 9 am UTC. As for the BTC options data, our collection starts from 2019-05-30 to 2023-10-04. The options data consists of 9 am UTC snapshots of the bid and ask IVs of call and put options. We chose to capture these snapshots at 9 am UTC to minimise the impact of market fluctuations that often occur due to the introduction of new options and futures listings at 8 am UTC.
 
We split the data into two datasets based on the availability of one-month option data. The first dataset is data observed from 2015-12-31 to 2019-05-29, where 2019-05-30 is the earliest date when 1-month options became available on Deribit. The second dataset is data observed from 2019-05-30 onward.

We design the empirical study as follows.
 Initially, we use BTC price data from the first dataset for the estimation of model parameters that are associated with jumps, as described in Section \ref{sec:two_stage}.
 In this stage, we apply POT estimation to filter out arrival times and magnitudes of positive and negative jumps in the BTC price in the testing data.
 In this empirical study, the training data are the daily BTC price observed from 2015-12-31 to 2019-05-29.
From the time series of BTC price, we compute the jump intensities for the training data, and then generalise the computation of intensities to the testing data based on the estimated model parameters.
 The model parameters are held fixed to infer jumps intensities to enforce no look-ahead bias in our empirical results.  
 
Once a set of model parameters is estimated and the time series of jump intensity rates are inferred, we calibrate model volatility $\sigma$, and risk-premium parameters $\chi^+$ and $\chi^-$ using options data in the second dataset.
Options data include mid quotes of implied volatilities of call and put options, specifically options with maturities that are closest to, but less than one month, and have a moneyness range between 0.8 and 1.2. These options are characterised by high open interest and high transaction volume.
 
The final step involves the analysis of the jump risk premia computed using Eqs \ref{eq:gamma_p} and \ref{eq:gamma_m}. Section \ref{sec:jumps_premia_evolution} presents the evolution of jump risk premia together with the main events in the BTC market using equations. Section \ref{sec:cost_of_carry} reports the relationship between the jump risk premia and the cost-of-carry of BTC futures. Section \ref{sec:PnL_realisation} shows application of the jump risk premia to profit and loss of delta-hedged option strategies.

\subsection{Goodness-of-fit}\label{sec:gof}
We visually assess the goodness of fit of our model to the BTC price dynamics through Q-Q plots of the inter-arrival times of unit rate Poisson processes versus the transformed inter-arrival times of observed BTC price jumps, 
 $\int_{T^+_{[k-1]}}^{T^+_{[k]}}\lambda^+_s \text{d}s$ for $k = 1,2,...,N^{(1)}$ and $\int_{T^-_{[k-1]}}^{T^-_{[k]}}\lambda^-_s \text{d}s$ for $k = 1,2,..., N^{(2)}$, where $N^{(1)}$ and $N^{(2)}$ are the numbers of positive and negative jumps observed until time $T$. 
Appendix \ref{sec:QQplot} provides further details of this approach. 

Figure \ref{fig:QQPlot_lambdas} presents the Q-Q plot of the model estimates. Model parameters are estimated from 2015-12-31 to 2019-05-29, as mentioned above. These estimates are applied to transform the jumps arrival times of the whole dataset from 2015-12-31 to 2023-10-04. In the Q-Q plots, data points from the first dataset are marked by circles, data points from the testing data are marked by crosses. The model estimates fit reasonably well for most of the negative jumps inter-arrival times. 
 
\begin{figure}
    \caption{
        The left panel shows the Q–Q plot of compensator-transformed interarrival times for positive jumps; the right panel shows the same for negative jumps. Crosses and circles represent jumps occurring before (in-sample) and after (out-of-sample) 2019-05-30, respectively. Under the model, the transformed interarrival times should follow a standard exponential distribution. The model is calibrated on data before 2019-05-30, and the estimated parameters are used to infer intensities for periods after 2019-05-30. The close agreement between theoretical and empirical quantiles indicates strong in- and out-of-sample performance.
\href{https://github.com/QuantLet/HJP/tree/main/codes/jupyterNotebooks/goodness_of_fit_P}{\includegraphics[height=\baselineskip]{qletlogo_tr.png}goodness\_of\_fit\_P}
    }
    \label{fig:QQPlot_lambdas}
    \begin{center}
    \includegraphics[width=\textwidth]{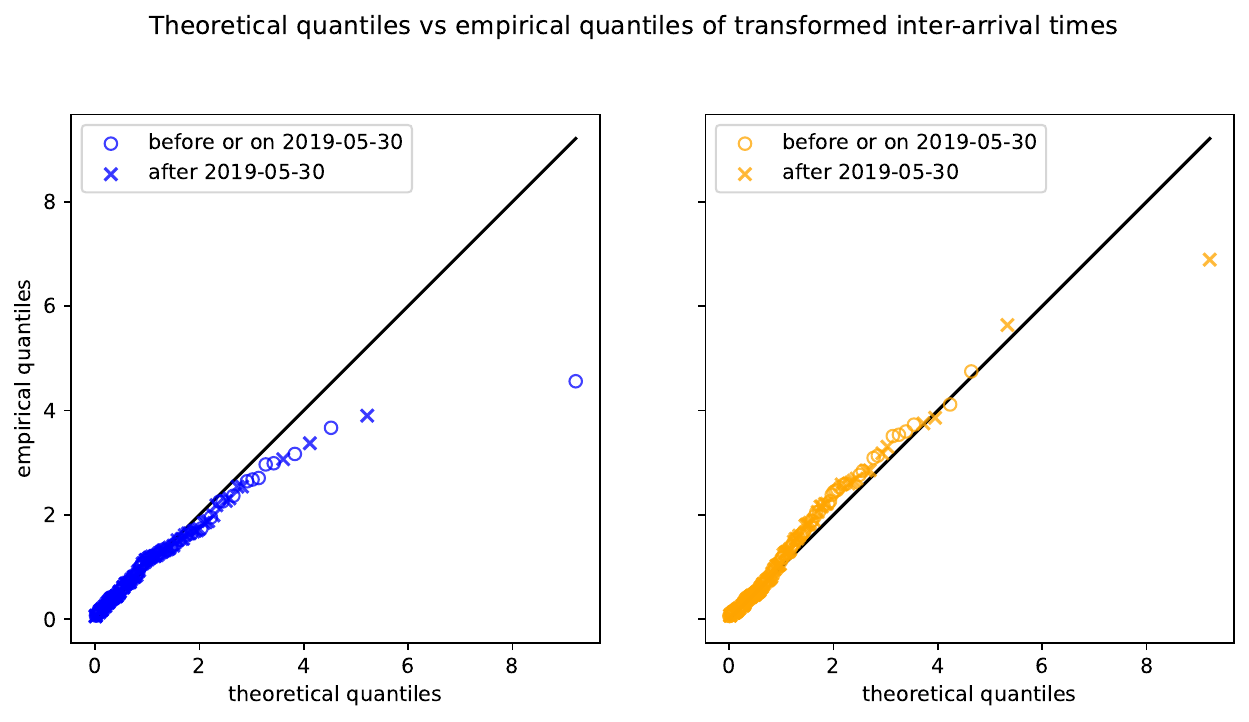}
    \end{center}
\end{figure}

\subsection{Evolution of jump risk  premia} \label{sec:jumps_premia_evolution}

The dynamics of jump risk premia, computed using Equations (\ref{eq:gamma_p}) and (\ref{eq:gamma_m}), from 2019-05-30 to 2023-10-04 are illustrated in Figure \ref{fig:jps}.
Overall, the jump risk premia vary strongly over time, while showing a tendency to mean-revert to a mean close to zero. In addition,
the negative jump risk premium fluctuates in the same direction as the positive jump risk premium, but with a stronger magnitude. This co-movement suggests the time-varying skewness observed in option prices. When jump premia are negative and put options have a higher demand than call options, the implied volatility skew becomes left-skewed compared to the skew inferred by the statistical measure, and vice versa.

Although \cite{Bollerslev2011}'s data and methodology differ significantly from ours, their findings are in line with our results. Using high-frequency data and a model-free approach to decompose volatility into continuous and jump components, they found that investors demand a significantly higher premium for downside jump risk than for upside jump risk. This asymmetry supports our observation that negative jump risk premia tend to have a greater magnitude than positive premia. When pricing longer-dated options, the extremely high jump risk premia are not expected to last very long due to the mean reversion of jump intensities.

A significant peak in premia was observed on 2021-01-16, building on the increase that started in early October 2020, a period marked by heightened institutional activity and a series of all-time high break-outs for BTC price.
The troughs and the market situations during the time include the following (i) 2020-03-22: COVID-19 outbreak starting in March 2020; (ii) 2021-05-27: environmental concerns over Bitcoin mining; (iii) 2022-06-20: $40\%$ decline in BTC caused by hawkins interest rate policy of the U.S. Federal Reserve amid high inflation rates; (iv) 2022-11-15: the FTX crisis and subsequent crash of cryptocurrencies.

\begin{figure}
    \caption{
        The time series plots illustrate the positive and negative jump risk premia, depicted by light blue and light orange lines, respectively.
        The darker lines in the plots represent the seven-day moving averages of the corresponding jump risk premia.
        These moving averages reveal a more distinct pattern in the evolution of the jump risk premia.
        The jump risk premia are estimated from 2019-05-30 to 2023-10-04. 
        \href{https://github.com/QuantLet/HJP/tree/main/codes/jupyterNotebooks/jumps_premia_evolution}{\includegraphics[height=\baselineskip]{qletlogo_tr.png}jumps\_premia\_evolution}
    }
\includegraphics[width=\textwidth]{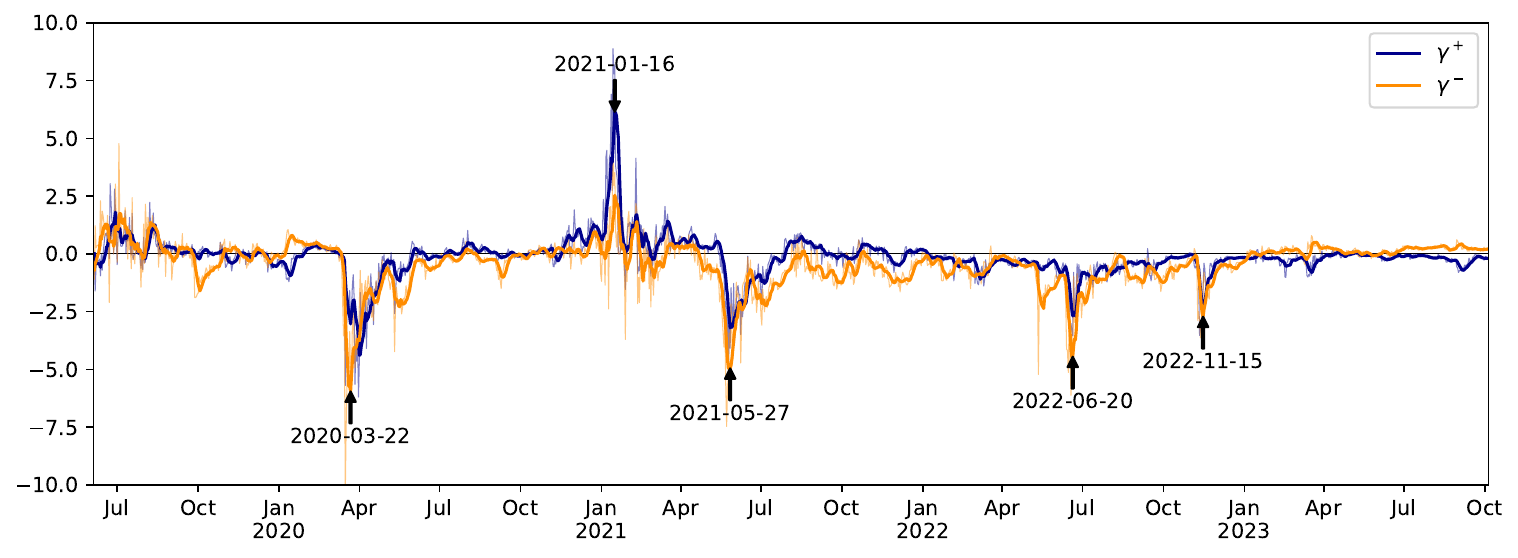}
\label{fig:jps}
\end{figure}

\subsection{Jump risk premia as factors in explaining BTC futures cost-of-carry}
\label{sec:cost_of_carry}

BTC futures are a popular choice for investors seeking leveraged exposures to BTC dynamics. Here, we investigate whether jump risk premia can serve as explanatory factors for BTC futures pricing, specifically through their relationship with the cost-of-carry.
The variable of interest is the BTC basis or cost-of-carry specified as follows:
\begin{align}
    c_t = \ln\big(F_t(\tau)/S_t\big)/\tau,
\end{align}
where $F_t(\tau)$ is the price at time $t$ of a futures contract with
time-to-maturity $\tau$, and $S_t$ is the spot price of BTC.
Figure \ref{fig:BTCfutures} shows the basis $c_t$, computed using the prevailing futures contract with a time-to-maturity just below one month, based on data from Deribit.
The cost-of-carry is annualized and exhibits high volatility, ranging from $-20\%$ to $80\%$. High values of the cost-of-carry arise when investors seek leverage. A trader can monetise a high value of the cost-of-carry by selling short the futures and going long BTC through a spot market.  

\begin{figure}[!t]
    \caption{
        The time series plot shows the BTC futures basis $c_t$ observed in the market. 
        The data is calculated using the price of the future with maturity $\tau$ less than, but closest to one month. 
        \href{https://github.com/QuantLet/HJP/tree/main/codes/jupyterNotebooks/cost_of_carry_factors}{\includegraphics[height=\baselineskip]{qletlogo_tr.png}cost\_of\_carry\_factors}
    }
    \includegraphics[width=\textwidth]{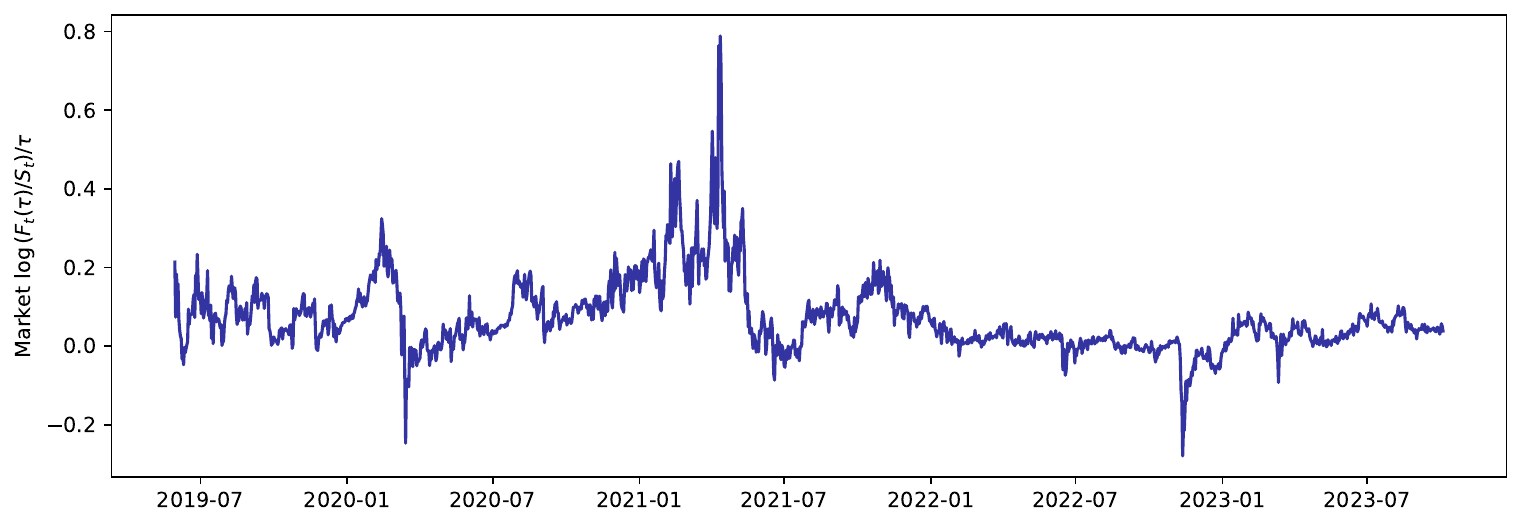}
    \label{fig:BTCfutures}
\end{figure}

In addition to the jump risk premia, the number of positive and negative jumps in the previous week, 
abbreviated respectively as $N^{(1)}_t - N^{(1)}_{t-7}$ and $N^{(2)}_t - N^{(2)}_{t-7}$, and the average funding rate of the week from Deribit BTC perpetual futures, abbreviated by $r^{\text{funding}}_{t, t-7}$, are included as independent variables in order to support the analysis. 
The regression of interest is given by
\begin{align}
    \Delta c_t &=  \beta_1 \Delta \gamma^+_t + \beta_2 \Delta \gamma^-_t 
    + \beta_3 \Delta \big(N^{(1)}_t - N^{(1)}_{t-7}\big) 
    + \beta_4 \Delta \big(N^{(2)}_t - N^{(2)}_{t-7}\big) 
    + \beta_5 \ \Delta r^{\text{funding}}_{t, t-7} + \Delta \varepsilon_t. \label{eq:regression_of_futures}
\end{align}

All variables are differenced to remove potential fixed effects and to mitigate the impact of any fixed effects. 
To mitigate the noise arising from daily fluctuations and the introduction of new contracts, the data selected for the regression analysis consist of market snapshots captured every Friday at 9 am UTC.
This timing is chosen because traders typically adjust their positions on Fridays to manage risk for the upcoming weekend.
Additionally, newly issued options and futures contracts typically begin trading at 8 am, allowing the 9 am UTC snapshot to incorporate the initial market responses of these new instruments. 
The variables are observed every Friday at 9 am UTC from 2019-06-14 to 2023-02-17.
Table \ref{tab:BTCfutures_regression} reports the estimated coefficients, Newey-West robustified $t$-statistics, and adjusted $R^2$ coefficients of regressions of $c_t$ on combinations of explanatory variables.

\begin{table}[!t]
    \caption{
        The table reports estimated coefficients, Newey-West robustified $t$-statistics,
        and adjusted $R^2$ coefficients of regressions of the futures basis $c_t=\ln(F_t(\tau)/S_t)/\tau$.
        The futures contract with time-to-maturity $\tau$ smaller than, but closest to one month is chosen for each observation time to represent the futures market. 
        The regressors are defined as follows: $\gamma^+_t$ and $\gamma^-_t$ are the jump risk premium. 
        $N^{(1)}_t - N^{(1)}_{t-7}$ and $N^{(2)}_t - N^{(2)}_{t-7}$ are number of positive and negative jumps in the week prior to $t$, respectively.
        $r^{\text{funding}}_{t, t-7}$ is the average funding rate of the week from Deribit BTC perpetual futures. 
        The Deribit BTC funding rate is chosen to represent the convenience yield in the crypto market. 
        All variables enter as first differences. 
         $^{*}$, $^{**}$, and $^{***}$ indicate statistical significance at the 10\%, 5\%, and 1\% levels, respectively.
        The variables are observed every Friday at 9 am UTC from 2019-06-14 to 2023-02-17.
        \href{https://github.com/QuantLet/HJP/tree/main/codes/jupyterNotebooks/cost_of_carry_factors}{\includegraphics[height=\baselineskip]{qletlogo_tr.png}cost\_of\_carry\_factors}
    }
    \begin{center}
      \resizebox{\textwidth}{!}{
\begin{tabular}{rSSSSSS}
\toprule
{} &            I &            II &          III &           IV &            V &           VI \\
\midrule
$\gamma^+_t$         &  0.0157$***$ &               &              &  0.0160$***$ &       0.0071 &    0.0075$*$ \\
            &     (2.6139) &               &              &     (3.5339) &     (1.2400) &     (1.7027) \\
$\gamma^-_t$         &   0.0132$**$ &               &              &       0.0089 &   0.0131$**$ &    0.0097$*$ \\
            &     (2.2257) &               &              &     (1.4316) &     (2.5760) &     (1.7995) \\
$N^{(1)}_t - N^{(1)}_{t-7}$ &              &    0.0102$**$ &              &  0.0112$***$ &              &    0.0073$*$ \\
            &              &      (2.1840) &              &     (3.0709) &              &     (1.8356) \\
$N^{(2)}_t - N^{(2)}_{t-7}$ &              &  -0.0179$***$ &              &  -0.0126$**$ &              &  -0.0114$**$ \\
            &              &     (-2.8429) &              &    (-2.4792) &              &    (-2.2707) \\
$r^{\text{funding}}_{t,t-7}$          &              &               &  0.0024$***$ &              &  0.0020$***$ &  0.0018$***$ \\
            &              &               &     (4.6951) &              &     (3.8976) &     (3.7371) \\
\midrule
Adj. $R^2$ & 0.1467 & 0.0866 & 0.2231 & 0.1956 & 0.2794 & 0.3107\\\bottomrule
\end{tabular}
    }
    \end{center}
    \label{tab:BTCfutures_regression}
\end{table}

The results of the regression analysis reveal a significant explanatory power of jump risk premia in determining the volatile cost-of-carry for BTC futures.
In Regression I, the coefficients for jump risk premia are both positive and statistically significant.
This model yields an adjusted $R^2$ of 14.7\%, together with the significant coefficients, supporting the hypothesis that futures prices are directly proportional to the jump risk premia.

The jump risk premia explanatory power is robust against the presence of the number of jumps and the funding rate variables. 
This conclusion is drawn by comparing the adjusted $R^2$ values and coefficients from Regression I with those from Regressions IV and V.
Combining the results of Regression I with that of Regression IV side by side, which includes the number of jumps from the previous week, the coefficient of the positive jump risk premium remains significant.
Inclusion of these variables improves the adjusted $R^2$ to 19.6\%. 
However, the coefficient of the negative jump risk premium decreases in magnitude and becomes statistically insignificant in this context.
A similar pattern is noted in Regression V.
When the funding rate is included in the Regression V, the coefficient of negative jump risk premium slightly diminishes while remaining significant, but
the coefficient of positive jump risk premium halves and becomes statistically insignificant. 
The inclusion of the average funding rate improves the adjusted $R^2$ to 28\%.
These improvements value of $R^2$ confirm that the number of jumps and the average funding rate provide additional information to explain the level of BTC futures cost-of-carry in addition to the jump risk premia. 
However, the shrinkage of the magnitudes and loss in statistical significance observed from the coefficients are probably due to the correlation between the variables.
Empirically, the correlation between the positive jump risk premium and the funding rate is around $36.7\%$; the correlation between the negative jump risk premium and the number of negative jumps is around $-24\%$.

The interpretation of the correlations is as follows. 
The positive correlation between the positive jump risk premium and the funding rate may be due to the increasing speculative demand for both call options and perpetual futures.
In contrast, the negative correlation between the negative jump risk premium and the number of negative jumps can result in higher put option prices during bearish periods with high volatility.
This interpretation is consistent with the findings of \cite{schmeling2023crypto}. 
They attribute high cost-of-carry of BTC futures to the interaction between the trend-chasing behaviour of market traders and the scarcity of arbitrage capital that conducts cash-and-carry trade.

Regression VI offers a comprehensive overview of how jump risk premia affect BTC futures prices, taking into account all relevant variables. 
This regression reaches an adjusted $R^2$ value of $31\%$, with each coefficient demonstrating statistical significance.
According to the regression result, a one-unit increase in the
positive jump risk premium is associated with a 0.75 basis point increase in the cost of carry of BTC futures $c_t$.
 Similarly, a one-unit increase in the negative jump risk premium corresponds to an increase of 0.97 basis points in $c_t$. 
The positive signs of both the positive and negative jump risk premia coefficients in Regression VI indicate that a positive change corresponds to (forward-looking) option and futures market being more optimistic than the price history suggests. 
The different sizes of the positive and negative jump risk premia coefficients can be attributed to risk aversion, as traders and investors are more sensitive to losses than to gains.

\subsection{Jump risk premia realisation}\label{sec:PnL_realisation}
In the spirit of \cite{carr1998towards} who quantifies the return variance risk premium on an asset using the market prices of options, we answer the question of how the risk premia can be realised. To this end, we investigate whether the jump risk premia defined in Equations~(\ref{eq:gamma_p}) and (\ref{eq:gamma_m}) can explain the one-week-ahead returns of delta-hedged short-option strategies, following a similar empirical approach to \cite{bakshi2003delta}. The option strategies are constructed using the same data and approach as \cite{lucicsepp2024}, where short positions in calls, puts, straddles and strangles are delta-hedged using Black-Scholes (BS) deltas and rebalanced hourly until expiration. 

All strategies are written on options expiring on the subsequent Friday at 8am UTC. 
As both positive and negative jump risk premia are observed, we define simple trading signals as follows: A positive  risk premium for positive jumps, which means that the market prices call options higher than implied by the historical measure, is considered an opportunity to sell call options and call spreads. 
The trading decision is similar for put options and put spreads when observing a negative risk premium on negative jumps. 
For portfolios involving straddle and strangles, which are combinations of call and put options, the trading strategy is to sell them if at least one of the risk premia is ``active'' in the sense described previously. 
The trading signals for all strategies are determined on each Friday at 9am UTC, and if active, the position is created (and delta-hedged with hourly rebalancing) until expiry the following Friday.\footnote{Long positions in calls and puts are not entered if the jump risk premia are ``reversed'' due to the loss in time value of the options (negative theta).}

Using \cite{bakshi2003delta}'s empirical methodology, we investigate the relationship between jump risk premia as defined above and one-week-ahead returns of the option strategy portfolios. 
This is achieved by conducting a regression against the jump risk premia of the one-week-ahead log returns $r_{t} = \ln(P_t/P_{t-1})$ of each strategy, with $P_t$ the option strategy value at time $t$. 
The regression of interest is
\begin{equation}
    r_{t+1} = \beta_0 + \beta_1 \gamma^+_t + \beta_2 \gamma^-_t %
     + \beta_3 \Delta \gamma^+_t + \beta_4 \Delta \gamma^-_t %
     + \beta_5 (N^{(1)}_t-N^{(1)}_{t-1}) 
    + \beta_6  (N^{(2)}_t-N^{(2)}_{t-1}),
    \label{regression_optionstrategies}
\end{equation}
where $r_{t+1}$ is the one-week-ahead options strategy return at time $t$ (expressed in basis points),
$\Delta$ is the first difference operator and 
$(N^{(1)}_t-N^{(1)}_{t-1})$ and $(N^{(2)}_t-N^{(2)}_{t-1})$ count the number of jumps that occured within the previous week.
The option strategies of interest are delta-hedged short-option strategies of different moneyness, as follows: 
\begin{enumerate}[(i)]
\item ATM: calls and puts with at-the-money strikes;
\item 10D: options with delta of +0.1 for calls and -0.1 for puts;
\item 25D: options with delta of +0.25 for calls and -0.25 for puts;
\item 25D spread: one unit of call with delta of +0.5 (-0.5 for put) and short two units of call with +0.25 delta (-0.25 delta for put);
\item ATM straddle: one unit of call and put each with ATM strikes;
\item 25D strangle: one unit of call with delta +0.25 and one unit of put with delta of -0.25;
\item 10D strangle: one unit of call with delta +0.1 and one unit of put with delta of -0.1.
\end{enumerate}
 
\begin{table}[!t]
    \caption{
        Regression results of one-week ahead portfolio returns (in bps) regressed on jump risk premia and jump realisations, cf.\ Eq.(\ref{regression_optionstrategies}), heteroskedasticity-robust $t$-statistics in parentheses. $*$, $**$, $***$ indicate statistical significance at the 10\%, 5\% and 1\% levels, respectively. Sample period: 2019-05-31 to 2023-09-29 (277 observations). 
    }
    \begin{center}
      \begin{tabular}{rSSSSSSSSSSS}
  \toprule	& \text{straddle\hspace*{-1cm}}	& \text{10d strangle\hspace*{-1cm}}	& \text{25d strangle\hspace*{-1cm}}	& \text{call spread\hspace*{-1cm}}	& \text{put spread\hspace*{-1cm}}	\\\midrule
const	& -1.985	& -0.898	& 3.827	& 7.194	& -2.822	\\
	& (-0.113)	& (-0.106)	& (0.31)	& (0.854)	& (-0.388)	\\
$\gamma^+_t$	& 22.652	& {\bf 23.889$**$}	& 26.59	& {\bf 24.41$**$}	& 0.309	\\
	& (0.968)	& (2.15)	& (1.535)	& (2.254)	& (0.027)	\\
$\gamma^-_t$	& {\bf -53.874$**$}	& {\bf -33.475$**$}	& {\bf -42.253$**$}	& -15.31	& -7.908	\\
	& (-2.13)	& (-2.541)	& (-2.36)	& (-1.25)	& (-0.662)	\\
$\Delta\gamma^+_t$	& 1.646	& {\bf -27.876$*$}	& -23.075	& {\bf -28.187$**$}	& {\bf -12.458$*$}\\
	& (0.093)	& (-1.787)	& (-1.458)	& (-2.059)	& (-1.784)	\\
$\Delta\gamma^-_t$	& 5.872	& {\bf 35.705$*$}	& 19.97	& 18.348	& 3.853	\\
	& (0.277)	& (1.911)	& (1.133)	& (1.066)	& (0.253)	\\
$N^{(1)}_t - N^{(1)}_{t-1}$	& 8.728	& 3.249	& -0.518	& -3.749	& 3.211	\\
	& (0.382)	& (0.325)	& (-0.031)	& (-0.47)	& (0.284)	\\
$N^{(2)}_t - N^{(2)}_{t-1}$	& {\bf 41.155$*$}	& 7.134	& 19.118	& 1.532	& -1.201	\\
	& (1.733)	& (0.632)	& (1.055)	& (0.299)	& (-0.122)	\\\midrule
Adj.\ R$^2$	& 0.049	& 0.045	& 0.03	& 0.044	& -0.002	\\\bottomrule
\end{tabular}
\bigskip

\begin{tabular}{rSSSSSSSSSSS}
\toprule	& \text{atm call\hspace*{-1cm}}	& \text{10d call\hspace*{-1cm}}	& \text{25d call\hspace*{-1cm}}	& \text{atm put\hspace*{-1cm}}	& \text{10d put\hspace*{-1cm}}	& \text{25d put\hspace*{-1cm}}	\\\midrule
const	& 3.038	& 0.75	& 4.182	& 2.222	& 0.784	& 1.459	\\
	& (0.46)	& (0.148)	& (0.769)	& (0.342)	& (0.199)	& (0.257)	\\
$\gamma^+_t$	& {\bf 19.367$**$}	& {\bf 20.446$*\!*\!*$}	& {\bf 22.365$*\!*\!*$}	& 2.931	& 2.524	& 3.992	\\
	& (2.234)	& (3.049)	& (3.024)	& (0.324)	& (0.461)	& (0.491)	\\
$\gamma^-_t$	& -11.39	& {\bf -15.932$*$}	& {\bf -14.445$*$}	& {\bf -22.474$**$}	& {\bf -11.587$**$}	& {\bf -16.64$*$}	\\
	& (-1.277)	& (-1.744)	& (-1.78)	& (-2.074)	& (-2.23)	& (-1.793)	\\
$\Delta\gamma^+_t$	& -7.976	& {\bf -22.573$*$}	& {\bf -19.624$**$}	& -1.039	& {\bf -7.417$**$}	& {\bf -8.926$*$}	\\
	& (-1.238)	& (-1.79)	& (-2.267)	& (-0.13)	& (-2.035)	& (-1.735)	\\
$\Delta\gamma^-_t$	& 10.594	& 22.222	& 15.607	& -0.67	& {\bf 7.178$**$}	& 3.187	\\
	& (1.513)	& (1.515)	& (1.53)	& (-0.071)	& (2.244)	& (0.512)	\\
$N^{(1)}_t - N^{(1)}_{t-1}$	& -11.773	& -1.281	& {\bf -7.385$*$}	& 11.508	& {\bf 7.74$**$}	& 6.15	\\
	& (-1.605)	& (-0.532)	& (-1.652)	& (1.272)	& (2.047)	& (0.958)	\\
$N^{(2)}_t - N^{(2)}_{t-1}$	& {\bf 14.153$*$}	& 3.682	& 7.496	& 5.904	& -1.132	& 2.745	\\
	& (1.714)	& (1.385)	& (1.53)	& (0.733)	& (-0.221)	& (0.406)	\\\midrule
Adj.\ R$^2$	& 0.041	& 0.097	& 0.071	& 0.048	& 0.029	& 0.039	\\\bottomrule
\end{tabular}
    \end{center}
    \label{tab:short_option_regressions}
\end{table}

Table \ref{tab:short_option_regressions} shows the results of the regression on one-week-ahead returns (in basis points), with robust (HAC) $t$-statistics in parentheses. The number of positive and negative jumps counted during the week are included as control variables. Recalling that a negative jump risk premium on downward jumps indicates high option prices relative to the statistical history (cf. Figure \ref{fig:IV_smiles_varying_gammas}, the mixed strategies (straddle, strangles) earn jump risk premia from downward jumps. The 10-delta strangle also earns risk premia from upward jumps as well as changes in jump risk premia. Here, a negative first difference in the upward jump risk premium indicates that the jump risk premium decreased over the course of the previous week, i.e., pricing of jump risk of out-of-the-money call options aligns better with historical jump risk. The negative coefficient persists when replacing $\Delta\gamma^+_t$ with the residuals of $\Delta\gamma_t^+$ regressed on $\gamma_t^+$, which rules out multicollinearity effects (likewise for $\Delta\gamma_t^-$ and $\gamma_t^-$). Hence, $\Delta\gamma_t^+$ and $\Delta\gamma_t^-$, when statistically significant, carry independent predictive performance driving the P\&L. A plausible interpretation of the negative coefficient associated with $\Delta\gamma_t^+$ is that this points to a calming market (expecting less right-tail events), which in turn makes the delta-hedge more robust and allows to earn the time-decay $\theta$ inherent in the option positions. The argument is similar for $\Delta\gamma_t^-$ and the positive coefficient. 

Overall, on average the short call option strategies earn the positive jump risk premium of upward jumps and the short put option strategies earn the negative jump risk premium of downward jumps. One has to bear in mind, though, that many other factors influence the weekly returns as the adjusted $R^2$ measures are fairly low for all strategies. 

The relative pricing interpretation also aligns with the analysis done by \cite{bakshi2003delta}. 
They demonstrate in Eq.(24) of their study that, theoretically, a more pronounced left tail in the underlying price distribution under the risk-neutral measure, as opposed to the ``real'' measure, enhances the performance of delta-hedged portfolios\footnote{\cite{bakshi2003delta} approach the construction of delta-hedged portfolios from a direction opposite to ours. While they engage in strategies that involve buying options and selling the underlying assets, we adopt the reverse strategy, selling options and buying the underlying.}.

\section{Conclusion}\label{sec:conclusion-and-outlook}

The cryptocurrency market provides a challenging venue for exploring jump risk premia due to its high volatility and sentiment-driven price dynamics. The dynamics of implied volatilities of cryptocurrency options exhibit varying skewness preferences. During bullish periods, the demand for call options increases and the implied skew becomes positive along with high carry costs for perpetual futures. During bearish periods, the demand for puts rises and the implied skew turns negative along with low or negative carry costs for perpetual futures. A similar regime-conditional behaviour of implied skew is present in options on the so-called ``meme'' stocks, which are characterised by high volumes in options driven by retail investors.

For modeling of regime-conditional implied volatility skew under the risk-neutral measure and skeweness of returns under the statistical measure, we have introduced a pricing model using Hawkes processes that incorporates the clustering of positive and negative jumps. We have further included two risk-premium parameters to model specification for modeling of regime-conditional risk-preferences. Here, one parameter is for the external specification of risk-premium for positive jumps, and one parameter is for risk-premium of negative jumps. Both parameters are inferred from the option prices observed on a given date.

For empirical investigation using our model, we have used options on Bitcoin (BTC). We have developed a sequential estimation procedure by first estimating model parameters under the statistical $\mathbb{P}$ measure and then by inferring two risk-premium parameters from options data under the risk-neutral $\mathbb{Q}$ measure. We have defined risk premia for positive and negative jumps as the difference between expected jumps under the risk-neutral and statistical probability measures, respectively. The implications of the jump risk premia estimated using our model are threefold. 

First, the estimated jump-risk premia demonstrate preferences for skewness risk observed in implied volatilities. Specifically for BTC markets, the risk premium of positive jumps exceeded that of negative jumps on several occasions during a few strong bullish periods in BTC. In contrast, the risk premium of negative jumps dominated during bearish periods in BTC. A strong premium for positive jumps indicates that market participants significantly overpay for expected upside during bull periods by buying call options. Vice versa, a significant premium of negative jumps is associated with increased demand for put options during bear periods. 

Second, jump risk premia have some explanatory power (with $R^2$ of $15\%$) to predict the dynamics of the cost-of-carry inferred from BTC futures. 
A regression model with the recent number of jumps and BTC perpetual funding rates as control variables yields an $R^2$ of $31\%$, with all coefficients being significant. The findings point to interdependencies between the BTC spot, options, and futures markets.

Third, we show that jump risk premia have explanatory power for one-week-ahead profit-and-loss (P\&L) of delta-hedged option strategies and portfolios. Simple decision strategies to sell short call options if the risk premium for upward jumps is positive, to sell put option if the risk premium for downward jumps is negative, or to sell straddles if either of the jump risk premia is active, can -- on average -- earn the respective risk premium. In all cases, the options are hedged with hourly rebalancing. 

Finally, we note that the regime-conditional behaviour of the implied option skew is also characteristic of the so-called ``meme'' stocks, which exhibit positive and negative implied volatility skews dependent on market conditions. 
In addition, the G-7 currencies tend to exhibit a variation in the sign of implied skew dependent on macro conditions. 
As a result, our framework can also be applied for analysis of the implied volatility dynamics of ``meme'' stocks and the G-7 currencies, which we leave for future research. 
In addition, we assume that the risk-premium parameters are constant for the specification of model dynamics. 
We leave an extension for time-varying dynamics of model risk-premium parameters as another potential research topic.

\bibliography{finance} %

\newpage
\appendix

\section{Proofs}\label{sec:proofs}

\subsection{Conditions for finite expected intensities}\label{sec:finite_expected_intensities}
The expected values of intensities satisfy the following system of ODEs.
For $s<t$, 
\begin{align*}
\frac{\text{d}\E(\lambda^+_t|\mathcal{F}_s)}{\text{d}t} &= \kappa^+(\theta^+ - \E(\lambda^+_t|\mathcal{F}_s)) + \beta_{11}\E(\lambda^+_t|\mathcal{F}_s)\E J^+
+ \beta_{12}\E(\lambda^-_t|\mathcal{F}_s)\E J^-,\ \E(\lambda^+_s|\mathcal{F}_s)=\lambda^+_s\\
\frac{\text{d}\E(\lambda^-_t|\mathcal{F}_s)}{\text{d}t} &= \kappa^-(\theta^- - \E(\lambda^-_t|\mathcal{F}_s)) + \beta_{21}\E(\lambda^+_t|\mathcal{F}_s)\E J^+
+ \beta_{22}\E(\lambda^-_t|\mathcal{F}_s)\E J^-,\ \E(\lambda^-_s|\mathcal{F}_s)=\lambda^-_s.
\end{align*}
See \cite{Errais2010} for proofs. 
Write the above system of ODEs in matrix form: 
For $s<t$, 
\begin{align}
\begin{pmatrix}
  \frac{\text{d}\E(\lambda^+_t|\mathcal{F}_s)}{\text{d}t} \\
  \frac{\text{d}\E(\lambda^-_t|\mathcal{F}_s)}{\text{d}t} 
\end{pmatrix}
= \Phi \begin{pmatrix}
  \E(\lambda^+_t|\mathcal{F}_s)\\
  \E(\lambda^-_t|\mathcal{F}_s)
\end{pmatrix} +C,\ 
\begin{pmatrix}
  \E(\lambda^+_s|\mathcal{F}_s)\\
\E(\lambda^-_s|\mathcal{F}_s)
\end{pmatrix}=
\begin{pmatrix}
\lambda^+_s\\
\lambda^-_s
\end{pmatrix}
\label{eq:lambda_ODEs}
\end{align}
where 
\begin{align*}
\Phi \overset{\text{def}}= 
\begin{pmatrix}
-\kappa^+ + \beta_{11}\E J^+ & \beta_{12}\E J^- \\
 \beta_{22}\E J^- & -\kappa^- + \beta_{21}\E J^+  
\end{pmatrix} \text{ and }
C \overset{\text{def}}= 
\begin{pmatrix}
\kappa^+\theta^+ \\ 
\kappa^-\theta^- 
  \end{pmatrix}.
\end{align*}
If $\Phi$ is invertible, the solution to Eq.(\ref{eq:lambda_ODEs}) is 
\begin{align*}
\begin{pmatrix}
  \E(\lambda^+_s|\mathcal{F}_s)\\
\E(\lambda^-_s|\mathcal{F}_s)
\end{pmatrix}=
\text{expm}(\Phi t) 
\begin{pmatrix}
  \lambda^+_s\\
  \lambda^-_s
\end{pmatrix}
+ \Phi^{-1}\left(\text{expm}(\Phi t)-I\right)C,
\end{align*}
where $\text{expm}()$ is the matrix exponential, 
$\Phi^{-1}$ is the matrix inverse of $\Phi$, and $I$ is a $2\times 2$ identity matrix.
The matrix exponential can be computed via eigenvectors and eigenvalues.
Let $S \Lambda S^{-1}$ be the eigendecomposition of $\Phi$, 
where $S$ is the full matrix of eigenvectors and $\Lambda = \text{diag}(\lambda_1,\ \lambda_2)$ is the diagonal matrix of eigenvalues of $\Phi$, $\lambda_1$ and $\lambda_2$;
By the definition of matrix exponential as a sum of powers, 
\begin{equation*}
  \text{expm}(\Phi t) = \sum_{n=0}^\infty \frac{(\Phi t)^n}{n!}
                      = \sum_{n=0}^\infty \frac{(S\Lambda S^{-1}t)^n}{n!}
                      = S\ \text{expm}(\Lambda t) S^{-1},
\end{equation*}
where $\text{expm}(\Lambda t)$ is simply $\text{diag}(e^{\lambda_1 t},\ e^{\lambda_2 t})$.
It is clear that if both $\lambda_1$ and $\lambda_2$ are negative and real (such that $\lim_{t \rightarrow \infty}\text{expm}(\Phi t))=0$), the asymptotic expected values of intensities is 
\begin{align*}
\lim_{t \rightarrow \infty} 
\begin{pmatrix}
  \E(\lambda^+_t|\mathcal{F}_s)\\
  \E(\lambda^-_t|\mathcal{F}_s)
\end{pmatrix}
= -\Phi^{-1}C.
\end{align*}
%
%
By the Gershgorin circle theorem, the eigenvalues of $\Phi$ lie within 
\begin{align*}
\left\{z:\left|z - (-\kappa^+ + \beta_{11}\E J^+)\right| \leq \left|\beta_{12}\E J^- \right| \right\} \cup 
\left\{z:\left|z - (-\kappa^- + \beta_{22}\E J^-)\right| \leq \left|\beta_{21}\E J^+ \right| \right\}.
\end{align*}
Since $\beta_{12}\E J^-$ and $\beta_{21}\E J^+$ are always positive, 
the eigenvalues of $\Phi$ lie within
\begin{align*}
  \left\{z:\left|z - (-\kappa^+ + \beta_{11}\E J^+)\right| \leq \beta_{12}\E J^-\right\} \cup 
  \left\{z:\left|z - (-\kappa^- + \beta_{22}\E J^-)\right| \leq \beta_{21}\E J^+\right\}.
\end{align*}
Therefore, the sufficient conditions for the eigenvalues of $\Phi$ being negative are
\begin{align}
\kappa^+ \geq \beta_{11}\E J^+ + \beta_{12}\E J^- \text{ and }
\kappa^- \geq \beta_{21}\E J^+ + \beta_{22}\E J^-. 
\end{align}
This means that if the mean reverting rates $\kappa^+$ and $\kappa^-$ are larger than their corresponding sum of excitements from jumps which are average in size,
 the asymptotic expected values of intensities are finite.

\subsection{Proof of Proposition \ref{prop:MGF}} \label{proof:prop:MGF}
Let $G(t,X_t,\lambda_t^+,\lambda_t^-)=\E\left(\exp(\omega X_T + \omega^+\lambda^+_T + \omega^-\lambda^-_T)\big|X_t, \lambda^+_t, \lambda^-_t\right)$. 
    The process $(G(t,X_t, \lambda_t^+, \lambda_t^-))_{t\geq 0}$ is a (Doob-L\'evy) martingale (see e.g.\ Theorem 2.31 of \cite{Klebaner2005}).
    By the predictable It\^o formula (e.g.\ Sections 10.2 and 8.4.3 of \citep{jeanblanc2009mathematical})
    and the martingale property of $G$, 
    \begin{align} \label{eq:drift_G}
        0 &=  G_t \nonumber\\
          &+ \left(\mu - \frac{\sigma^2}{2}-\lambda^+_t\E(e^{J^+}-1)-\lambda^-_t\E(e^{J^-}-1\right)G_x 
          +\frac{\sigma^2}{2}G_{xx} \nonumber\\
          &+\kappa^+(\theta^+-\lambda^+_t)G_{\lambda^+} +\kappa^-(\theta^--\lambda^-_t)G_{\lambda^-} \nonumber\\
          &+\lambda^+_t\left[\int_\mathbb{R}G(t, X_t+j^+, \lambda_t^++\beta_{11}j^+, \lambda_t^-+\beta_{21}j^+) \varpi^+(j^+)\dd j^+ - G(t,X_t, \lambda_t^+, \lambda_t^-)\right]\nonumber\\
          &+\lambda^-_t\left[\int_\mathbb{R}G(t,X_t+j^-, \lambda_t^++\beta_{12}j^-, \lambda_t^-+\beta_{22}j^-, T) \varpi^-(j^-)\dd j^- - G(t,X_t, \lambda_t^+, \lambda_t^-)\right], 
    \end{align}
    where $G_t, G_x, G_{xx}, G_{\lambda^+}, G_{\lambda^-}$ denote the partial derivatives of $G$. 
    Assume that $G$ has an exponential affine form (see e.g.\ \cite{Errais2010}):
    \begin{equation*}
        G(t, x, \lambda^+, \lambda^-) = \exp\left(A(t,T)+ B(t,T)x +C(t,T)\lambda^+ +D(t,T)\lambda^-\right),
    \end{equation*}
    where the functions $A,B,C,D$ are time dependent functions with terminal conditions: $A(T,T)=0$, $B(T,T)=\omega$, $C(T,T)=\omega^+$, and $D(T,T)=\omega^-$. 
    Inserting the partial derivatives into Equation~(\ref{eq:drift_G}),
    \begin{align*}
        0 &= \left(A_t + B_t X_t + C_t \lambda_t^+ + D_t\lambda_t^-\right)G \\
          &+ \left(\mu - \frac{\sigma^2}{2}-\lambda^+_t\E(e^{J^+}-1)-\lambda^-_t\E(e^{J^-}-1)\right)BG\\
          &+\frac{\sigma^2}{2}B^2G \\
          &+\kappa^+(\theta^+-\lambda^+_t)C +\kappa^-(\theta^--\lambda^-_t)DG \\
          &+\lambda^+_t\left[\int_\mathbb{R}G(t, X_t+j^+, \lambda_t^+ +\beta_{11}j^+, \lambda_t^- +\beta_{21}j^+) \varpi^+(j^+)\dd j^+ - G(t,X_t, \lambda_t^+, \lambda_t^-)\right]\\
          &+\lambda^-_t\left[\int_\mathbb{R}G(t, X_t+j^-, \lambda_t^+ +\beta_{12}j^-, \lambda_t^- +\beta_{22}j^-) \varpi^-(j^-)\dd j^- - G(t,X_t, \lambda_t^+, \lambda_t^-)\right].
    \end{align*}
    Grouping terms by the state variables and dividing both sides by $G$ (recall $G>0$) gives, 
    \begin{align*}
        0 &= A_t + \mu B + \frac{\sigma^2}{2}\left(B^2-B\right) + \kappa^+\theta^+C + \kappa^-\theta^-D\\
          &+ B_t X_t \\
          &+ \lambda_t^+\left[\mathcal{L}^{(+)}(-B-C\beta_{11}-D\beta_{21})-1+C_t - \E(e^{J^+}-1)B - \kappa^+C\right]\\
          &+ \lambda_t^-\left[\mathcal{L}^{(-)}(-B-C\beta_{12}-D\beta_{22})-1+D_t- \E(e^{J^-}-1)B - \kappa^-D\right],
    \end{align*}
    with $\mathcal{L}^{(+)}$ and $\mathcal{L}^{(-)}$ the MGFs of positive and negative jump sizes, cf.\ Equations~(\ref{eq:L_p}) and (\ref{eq:L_m}).  
    The above equation holds for all values of the state variables, $x, \lambda^+, \lambda^-$, so the following system of PDEs must hold
    \begin{align*}
        0 &= A_t + \mu B + \frac{\sigma^2}{2}\left(B^2-B\right) + \kappa^+\theta^+C + \kappa^-\theta^-D\\
        0 &= B_t  \\
        0 &= \mathcal{L}^{(+)}(-B-C\beta_{11}-D\beta_{21})-1+\partial_t C - \E(e^{J^+}-1)B - \kappa^+C\\
        0 &= \mathcal{L}^{(-)}(-B-C\beta_{12}-D\beta_{22})-1+\partial_t D - \E(e^{J^-}-1)B - \kappa^-D.
    \end{align*}
    The solution of the above system of PDEs is, for $t\geq 0$
    \begin{align*}
        A_t(t,T) &= - \mu \omega - \frac{\sigma^2}{2}\left(\omega^2-\omega\right) - \kappa^+\theta^+C(t,T) - \kappa^-\theta^-D(t,T)\\
        B_t(t,T) &= 0,\\
        C_t(t,T) &= -\mathcal{L}^{(+)}(-\omega-C(t,T)\beta_{11}-D(t,T)\beta_{21})+1 + \E(e^{J^+}-1)\omega + \kappa^+C(t,T),\\
        D_t(t,T) &= -\mathcal{L}^{(-)}(-\omega-C(t,T)\beta_{12}-D(t,T)\beta_{22})+1 + \E(e^{J^-}-1)\omega + \kappa^-D(t,T),
    \end{align*}
    where the terminal conditions are $A(T,T)=0$, $=B(T,T)=\omega$, $C(T,T)=\omega^+$, and $D(T,T)=\omega^-$.

\subsection{Proof of Proposition \ref{prop:martingaleConditions}}\label{proof:prop:martingaleConditions}  
    Let $m_t=\ln(M_t)$, so
    \begin{align*}
         m_t &= q_1^+(\xi^+) \lambda_t^+ + q_1^-(\xi^-) \lambda_t^-
          + \xi^+ \sum_{j=1}^{N_t^{(1)}} J_j^+ + \xi^- \sum_{j=1}^{N_t^{(2)}} J_j^- + 
            q_2(\xi^+, \xi^-) t  \\%
          &\phantom{=\,}-\frac{1}{2} \int_0^t\varphi(s)^2\,
                                    \dd s - \int_0^t \varphi(s)\, \dd  W_s
    \end{align*}
    with dynamics
    \begin{align*}
        \dd m_t &= \phantom{=} q_1^+(\xi^+) \kappa^+ (\theta^+ - \lambda_t^+)\, \dd
                t + (q_1^+(\xi^+)
                \beta_{11} + \xi^+) J^+\, \dd N_t^{(1)} + q_1^+(\xi^+)
                \beta_{12} J^-\, \dd N_t^{(2)}\\ %
        &\phantom{=\,} +  q_1^-(\xi^-) \kappa^- (\theta^- - \lambda_t^-)\, \dd
                t + q_1^-(\xi^-)
                \beta_{21} J^+\, \dd N_t^{(1)} + (q_1^-(\xi^-)
                \beta_{22} + \xi^-) J^-\, \dd N_t^{(2)} \\%
        & \phantom{=\,} + q_2(\xi^+,\xi^-)\, \dd t -
        \frac{1}{2}\varphi(t)^2\, \dd t - \varphi(t)\, \dd W_t\\
        &= \phantom{=} q_1^+(\xi^+) \kappa^+ (\theta^+ - \lambda_t^+)\, \dd
         t + \chi^+ J^+\, \dd N_t^{(1)}  %
        +  q_1^-(\xi^-) \kappa^- (\theta^- - \lambda_t^-)\, \dd
        t + \chi^- J^-\, \dd N_t^{(2)} \\%
        & \phantom{=\,} + q_2(\xi^+,\xi^-)\, \dd t - \frac{1}{2}\varphi(t)^2\, \dd t - \varphi(t)\, \dd W_t. 
    \end{align*}
    Denote the random measures of $J^\pm$ by $\Xi^\pm$, so that $J^+=\int_0^\infty \Xi^+(\dd z)$ and $J^-=\int_{-\infty}^0 \Xi^-(\dd z)$. Applying Ito's Lemma for semimartingales \citep[e.g.\ Section 8.10 of][]{Klebaner2005} to $M_t$ gives
    \begin{align*}
        \dd M_t &= M_{t-}\, \dd m_t + \frac{1}{2} M_t \, \dd[m_t, m_t] \\%
        &\phantom{=\,} + M_{t-} \int_0^\infty
                \left(\e^{\chi^+ z} - 1 - \chi^+ z\right) \Xi^+(\dd z)\, \dd 
                N_t^{(1)} %
        + M_{t-} \int_{-\infty}^0
                \left(\e^{\chi^- z} - 1 - \chi^- z\right) \Xi^-(\dd z)\, \dd N_t^{(2)}. 
    \end{align*}
    Expanding the $\dd m_t$-term and compensating the jumps gives
    \begin{align*}
        \dd M_t &= M_t (
                q_1^+(\xi^+) \kappa^+\theta^+ +
                q_1^-(\xi^-)
                \kappa^-\theta^- + q_2(\xi^+,\xi^-))\, \dd
                t %
                - M_t \varphi(t)\, \dd W_t\\ %
        &\phantom{=\,}- M_t\lambda_t^+ \left(q_1^+(\xi^+) \kappa^+ -
                \int_0^\infty
                \left(\e^{\chi^+ z} - 1\right)\, \varpi^+(\dd z)\,\right) \dd t \\
        &\phantom{=\,} - M_t \lambda_t^- \left(q_1^-(\xi^-) \kappa^- - 
                 \int_{-\infty}^0  \left(\e^{\chi^- z} - 1\right)\, \varpi^-(\dd
                z)\right) \dd t \\
        &\phantom{=\,} + M_{t-} \int_0^\infty  \left(\e^{\chi^+ z} - 1\right) (\Xi^+(\dd z)\, \dd N_t^{(1)} - \lambda_t^+
                \varpi^+(\dd z)\, \dd t)\\ %
        &\phantom{=\,} + M_{t-} \int_{-\infty}^0 \left(\e^{\chi^- z} - 1\right)\, (\Xi^-(\dd
                z)\, \dd N_t^{(2)} - \lambda_t^-\varpi^-(\dd z)\, \dd t).
    \end{align*}
    The drift terms vanish if the conditions hold.

\subsection{Proof of Proposition \ref{prop:HawkesUnderQ}}\label{proof:HawkesUnderQ}

    The initial guess of the relationship between $\lambda^+$ and $\lambda^{+, \mathbb{Q}}$, and $\lambda^-$ and $\lambda^{-, \mathbb{Q}}$ under the measure $\mathbb{P}$ is 
    \begin{align*}
         \lambda^{+, \mathbb{Q}}_t &= \mathcal{L}^{(+)}(-\chi^+)\lambda^+_t, \\
         \lambda^{-, \mathbb{Q}}_t &= \mathcal{L}^{(-)}(-\chi^-)\lambda^-_t, \quad t\geq 0.
    \end{align*}
    We verify the above by comparing the MGFs under the two measures. 
    Start with the MGF of jump intensities:
    \begin{align*}
        \E^\mathbb{Q}\left[\exp(\omega^+\lambda^{+, \mathbb{Q}}_T + \omega^-\lambda^{-, \mathbb{Q}}_T)\big| \mathcal{F}_t\right]
        = &e^{-m}\E\left(\exp(m_T+\omega^+\lambda^{+, \mathbb{Q}}_T+\omega^-\lambda^{-, \mathbb{Q}}_T)\big| \lambda^+_t=\lambda^+, \lambda^-_t=\lambda^-, m_t=m\right)\\
        \overset{\text{def}}= &e^{-m}H(t \omega^+, \omega^-; \lambda^+, \lambda^-, m, T).
    \end{align*}
    Focus on the form of $H$ that is comparable to the MGF under measure $\mathbb{P}$ stated in Proposition  \ref{prop:MGF}:
    \begin{align*}
        H(t, \lambda^+, \lambda^-, m, T) = \exp\left(A(t,T) + C(t,T)\mathcal{L}^{(+)}(-\chi^+)\lambda^+ + D(t,T)\mathcal{L}^{(-)}(-\chi^-)\lambda^- + E(t,T)m\right),
    \end{align*}
    where the functions $A, C$ and $D$ are defined in Proposition \ref{prop:MGF} and have the same terminal conditions; $E$ is a time-dependent function with terminal condition $E(T,T)=1$.
    Applying Ito's lemma and by martingale property of conditional expectation, 
    \begin{align*}
        0 &= H_t \\
        &+\left(\kappa^+_1(\xi^+)\kappa^+(\theta^+-\lambda^+)+\kappa^-_1(\xi^-)\kappa^-(\theta^--\lambda^-)+q_2(\xi^+,\xi^-)-\frac{1}{2}\varphi^2(t)\right)H_m
        +\frac{1}{2}\varphi^2(t)H_{mm}\\
        &+\kappa^+(\theta^+-\lambda^+) H_{\lambda^+}
        +\kappa^-(\theta^--\lambda^-) H_{\lambda^-}\\
        &+\lambda^+ \left[\int_\mathbb{R}H(t, \lambda^++\beta_{11}j^+, \lambda^-+\beta_{21}j^+, m+\chi^+j^+, T)\varpi^+(j^+)\dd j^+-H(t,  \lambda^+, \lambda^-, m, T)\right]\\
        &+\lambda^- \left[\int_\mathbb{R}H(t, \lambda^++\beta_{12}j^-, \lambda^-+\beta_{22}j^-, m+\chi^+j^-, T)\varpi^-(j^-)\dd j^--H(t, \lambda^+, \lambda^-, m, T)\right].
    \end{align*}
    By the first martingale condition of process $\left\{M_t\right\}_{t \geq 0}$, Eq.(\ref{eq:M1}), we simplify the above equation:
    \begin{align*}
        0 &= H_t \\
        &-\left(\kappa^+_1(\xi^+)\kappa^+\lambda^++\kappa^-_1(\xi^-)\kappa^-\lambda^-+\frac{1}{2}\varphi^2(t)\right)H_m
        +\frac{1}{2}\varphi^2(t)H_{mm}\\
        &+\kappa^+(\theta^+-\lambda^+)H_{\lambda^+}
        +\kappa^-(\theta^--\lambda^-)H_{\lambda^-}\\
        &+\lambda^+ \left[\int_\mathbb{R}H(t, \lambda^++\beta_{11}j^+, \lambda^-+\beta_{21}j^+, m+\chi^+j^+, T)\varpi^+(j^+)\dd j^+-H(t, \lambda^+, \lambda^-, m, T)\right]\\
        &+\lambda^- \left[\int_\mathbb{R}H(t, \lambda^++\beta_{12}j^-, \lambda^-+\beta_{22}j^-, m+\chi^-j^-, T)\varpi^-(j^-)\dd j^--H(t, \lambda^+, \lambda^-, m, T)\right].
    \end{align*}
    Inserting the partial derivatives of $H$ into the above equation and dividing both sides by $H$,
    \begin{align*}
        0 &= A_t + C_t \,\mathcal{L}^{(+)}(-\chi^+)\lambda^+ + D_t\, \mathcal{L}^{(-)}(-\chi^-) \lambda^-  + E_t\, m\\
        &-\left(\kappa^+_1(\xi^+)\kappa^+\lambda^++\kappa^-_1(\xi^-)\kappa^-\lambda^-+\frac{1}{2}\varphi^2(t)\right)E
        +\frac{1}{2}\varphi^2(t)E^2\\
        &+\kappa^+(\theta^+-\lambda^+)\mathcal{L}^{(+)}(-\chi^+)C
        +\kappa^-(\theta^--\lambda^-)\mathcal{L}^{(-)}(-\chi^-)D\\
        &+\lambda^+\left[\mathcal{L}^{(+)}(-C\mathcal{L}^{(+)}(-\chi^+)\beta_{11}-D\mathcal{L}^{(-)}(-\chi^-)\beta_{21}-E\chi^+)-1\right]\\
        &+\lambda^-\left[\mathcal{L}^{(-)}(-C\mathcal{L}^{(+)}(-\chi^+)\beta_{12}-D\mathcal{L}^{(-)}(-\chi^-)\beta_{22}-E\chi^-)-1\right]
    \end{align*}
    Grouping terms by the state variables, $x$, $\lambda^+$, and $\lambda^-$,
    \begin{align*}
        0 &= A_t + \frac{1}{2}\varphi^2(t)(E^2-E) +\kappa^+\theta^+\mathcal{L}^{(+)}(-\chi^+)C + \kappa^-\theta^-\mathcal{L}^{(-)}(-\chi^-)D\\
        &+ E_t\, m\\
        &+\lambda^+\left[\mathcal{L}^{(+)}(-C\mathcal{L}^{(+)}(-\chi^+)\beta_{11}-D\mathcal{L}^{(-)}(-\chi^-)\beta_{21}-E\chi^+)-1 + C_t\,  \mathcal{L}^{(+)}(-\chi^+)-\kappa^+_1(\xi^+)\kappa^+-\kappa^+C\mathcal{L}^{(+)}(-\chi^+)\right]\\
        &+\lambda^-\left[\mathcal{L}^{(-)}(-C\mathcal{L}^{(+)}(-\chi^+)\beta_{12}-D\mathcal{L}^{(-)}(-\chi^-)\beta_{22}-E\chi^-)-1 + D_t\, \mathcal{L}^{(-)}(-\chi^-)-\kappa^-_1(\xi^-)\kappa^--\kappa^-D\mathcal{L}^{(-)}(-\chi^-)\right]
     \end{align*}
    The above equation holds for all values of $x,\ \lambda^+$ and $\lambda^-$, we infer 
    \begin{align*}
        \partial_t A(t,T) &= -\kappa^+\theta^+\mathcal{L}^{(+)}(-\chi^+)C(t,T) - \kappa^-\theta^-\mathcal{L}^{(-)}(-\chi^-)D(t,T)\\
        \partial_t C(t,T) &= -\frac{\mathcal{L}^{(+)}(-C(t,T)\mathcal{L}^{(+)}(-\chi^+)\beta_{11} - D(t,T)\mathcal{L}^{(-)}(-\chi^-)\beta_{21}-\chi^+)}{\mathcal{L}^{(+)}(-\chi^+)}+\frac{1+\kappa^+_1(\xi^+)\kappa^+}{\mathcal{L}^{(+)}(-\chi^+)}+\kappa^+C(t,T),\\
        \partial_t D(t,T) &= -\frac{\mathcal{L}^{(-)}(-C(t,T)\mathcal{L}^{(+)}(-\chi^+)\beta_{12} - D(t,T)\mathcal{L}^{(-)}(-\chi^-)\beta_{22}-\chi^-)}{\mathcal{L}^{(-)}(-\chi^-)}+\frac{1+\kappa^-_1(\xi^-)\kappa^-}{\mathcal{L}^{(-)}(-\chi^-)}+\kappa^-D(t,T),\\
        \partial_t E(t,T) &= 0.
    \end{align*}
    The terminal conditions are $A(T,T)=0,\ C(T,T)=\omega^+,\ D(T,T)=\omega^-,\text{ and} E(T,T)=1$.

    The first term of $\partial_t C$, $\frac{\mathcal{L}^{(+)}(-C(t,T)\mathcal{L}^{(+)}(-\chi^+)\beta_{11} - D(t,T)\mathcal{L}^{(-)}(-\chi^-)\beta_{21}-\chi^+)}{\mathcal{L}^{(+)}(-\chi^+)}$, is in fact an Esscher transformation applied onto the positive jumps size density $\varpi^+$ with $-\chi^+$ as the Esscher parameter (A general form of Esscher transform can be found in Eq. (2.6) in \cite{gerber1995option}.). 
    Similarly, the first term of $\partial_t D$ is an Esscher transform to the negative jumps size density with $-\chi^-$ as the Esscher parameter.
    We denote the transformed MGFs of the jumps sizes as 
    \begin{align*}
        \mathcal{L}^{(+), \mathbb{Q}}(x) &\overset{\text{def}} = \frac{\mathcal{L}^{(+)}(x-\chi^+)}{\mathcal{L}^{(+)}(-\chi^+)} \text{, and }\\
        \mathcal{L}^{(-), \mathbb{Q}}(x) &\overset{\text{def}} = \frac{\mathcal{L}^{(-)}(x-\chi^-)}{\mathcal{L}^{(-)}(-\chi^-)}.
    \end{align*}
    On the other hand, the second terms of $\partial_t C$ and $\partial_t D$ can be simplified by the martingale conditions in Eq.(\ref{eq:M1}).

    To summarise,
    \begin{align*}
        \partial_t C(t,T) &= -\mathcal{L}^{(+),\mathbb{Q}}\left(-C(t,T)\mathcal{L}^{(+)}(-\chi^+)\beta_{11}-D(t,T)\mathcal{L}^{(-)}(-\chi^-)\beta_{21}\right)+1+\kappa^+C(t,T)\\
        \partial_t D(t,T) &= -\mathcal{L}^{(-),\mathbb{Q}}\left(-C(t,T)\mathcal{L}^{(+)}(-\chi^+)\beta_{12}-D(t,T)\mathcal{L}^{(-)}(-\chi^-)\beta_{22}\right)+1+\kappa^-D(t,T).
    \end{align*}
     
    Next, we work out the Esscher transformed densities of the jump sizes (see Eq.(2.5) of \cite{gerber1995option}),
    \begin{align*}
        \varpi^{+,\mathbb{Q}}(j) \overset{\text{def}}= \frac{e^{-\chi^+ j}\varpi^{+}(j)}{\int_\mathbb{R} e^{-\chi^+}\varpi^+(y) \text{d}y} &= \frac{1}{\eta^{+, \mathbb{Q}}}\exp\left(-\frac{1}{\eta^{+, \mathbb{Q}}}(j-\nu^+)\right)\mathbbm{1}_{\{j > \nu^+\}}\\
        \varpi^{-,\mathbb{Q}}(j) \overset{\text{def}}= \frac{e^{-\chi^- j}\varpi^{-}(j)}{\int_\mathbb{R} e^{-\chi^-}\varpi^-(y) \text{d}y}&= \frac{1}{\eta^{-, \mathbb{Q}}}\exp\left(\phantom{-}\frac{1}{\eta^{-, \mathbb{Q}}}(j-\nu^-)\right)\mathbbm{1}_{\{j < \nu^-\}},
    \end{align*}
    where 
    \begin{align*}
        \eta^{+,\mathbb{Q}} 
        = \frac{\eta^+}{1-\eta^+\chi^+}, \ \text{and}\
        \eta^{-,\mathbb{Q}} 
        = \frac{\eta^-}{1+\eta^-\chi^-}.
    \end{align*}
 
    To summarise the parameter changes, the MGF of the jump intensities under the measure $\mathbb{Q}$ is
    \begin{align*}
        \E^\mathbb{Q}\left(e^{\omega^+\lambda^+_T + \omega^-\lambda^-_T}|\mathcal{F}_t\right)
        &= \exp\left(A(t,T) + C(t,T)\mathcal{L}^{(+)}(-\chi^+)\lambda^+_t+ D(t,T)\mathcal{L}^{(-)}(-\chi^-)\lambda^-_t\right)\\
        &= \exp\left(A(t,T) + C(t,T)\lambda^{+, \mathbb{Q}}_t+ D(t,T)\lambda^{-, \mathbb{Q}}_t\right),
    \end{align*}
    where the function $A, C$ and $D$ solve the following system of PDEs
    \begin{align*}
        \partial_t A(t,T) &= -\kappa^+\theta^+\mathcal{L}^{(+)}(-\chi^+)C(t,T)-\kappa^-\theta^-\mathcal{L}^{(-)}(-\chi^-)D(t,T), &A(T,T)&=0\\
        \partial_t C(t,T) &= -\mathcal{L}^{(+),\mathbb{Q}}(-C(t,T)\mathcal{L}^{(+)}(-\chi^+)\beta_{11}-D(t,T)\mathcal{L}^{(-),\mathbb{Q}}(-\chi^-)\beta_{21})+1 + \kappa^+C(t,T),&C(T,T)&=\omega^+\\
        \partial_t C(t,T) &= -\mathcal{L}^{(-),\mathbb{Q}}(-C(t,T)\mathcal{L}^{(+)}(-\chi^+)\beta_{12}-D(t,T)\mathcal{L}^{(-),\mathbb{Q}}(-\chi^-)\beta_{22})+1 + \kappa^-D(t,T),&D(T,T)&=\omega^-.
    \end{align*}
    By comparing the above MGF under $\mathbb{Q}$ to the MGF under $\mathbb{P}$ in Proposition \ref{prop:MGF} (when $\omega=0$), we conclude that the measure change 
    \begin{itemize}
        \item alters the positive and negative jump intensities by multipliers $\mathcal{L}^{(+)}(-\chi^+)$ and $\mathcal{L}^{(-)}(-\chi^-)$, respectively. 
        \item shifts the rates of the positive and negative jump sizes from $\eta^+$ to $\eta^+/(1-\eta^+\chi^+)$ and from $\eta^-$ to $\eta^-/(1+\eta^-\chi^-)$, respectively.
        \item preserves the mean reversion rates of jump intensities $\kappa^+$ and $\kappa^-$ and the shift parameters of jump sizes $\nu^+$ and $\nu^-$.
    \end{itemize}

\subsection{Proof of Proposition \ref{prop:risk-neutralisation}}\label{proof:prop:risk-neutralisation}
    The proof relies on comparing the the MGFs of $X$ under $\mathbb{Q}$ and $\mathbb{P}$. 
    First, we derive the MGF of $X$ under $\mathbb{Q}$. Denote $m_t=\ln(M_t)$. 
    Given $X_t=x, \lambda_t^+=\lambda^+, \lambda_t^-=\lambda^-$ and $m_t=m$, the MGF is 
    \begin{align*}
        \E^\mathbb{Q}\left(e^{\omega X_T}\big|\mathcal{F}_t\right) &= \E\left(e^{m_T-m+\omega X_T}\big|\mathcal{F}_t\right)\\
        &= e^{-m}\E\left(e^{m_T+\omega X_T}\big|X_t=x, \lambda^+_t=\lambda^+, \lambda^-_t=\lambda^-, M_t=m,T\right),
    \end{align*}
    where $m_t=ln(M_t)$.
    Let $H(t \omega; x, \lambda^+, \lambda^-,m, T) = \E\left(e^{M_T+\omega X_T}\big|X_t=x, \lambda^+_t=\lambda^+, \lambda^-_t=\lambda^-, M_t=m\right)$. By the tower law, $\left\{H(t)\right\}_{t \geq 0}$ is a martingale.
    By Ito's lemma and martingale property of conditional expectation, 
    \begin{align*}
        0 &= \partial_t H \\
        &+ \left(\mu-\frac{\sigma^2}{2}-\lambda^+\E\left(e^{J^+}-1\right)-\lambda^-\E\left(e^{J^-}-1\right)\right)\partial_x H+\frac{\sigma^2}{2}\partial^2_x H\\
        &+\left(\kappa^+_1(\xi^+)\kappa^+(\theta^+-\lambda^+)+\kappa^-_1(\xi^-)\kappa^-(\theta^--\lambda^-)+q_2(\xi^+,\xi^-)-\frac{1}{2}\varphi^2(t)\right)\partial_m H
         +\frac{1}{2}\varphi^2(t)\partial^2_m H\\
        &-\sigma\varphi(t)\partial^2_{xm}\\ 
        &+\kappa^+(\theta^+-\lambda^+)\partial_+ H
         +\kappa^-(\theta^--\lambda^-)\partial_- H\\
        &+\lambda^+ \left[\int_\mathbb{R}H(t, \omega^+, \omega^-; \lambda^++\beta_{11}j^+, \lambda^-+\beta_{21}j^+, m+\chi^+j^+, T)\varpi^+(j^+)\dd j^+-H(t, \omega^+, \omega^-; \lambda^+, \lambda^-, m, T)\right]\\
        &+\lambda^- \left[\int_\mathbb{R}H(t, \omega^+, \omega^-; \lambda^++\beta_{12}j^-, \lambda^-+\beta_{22}j^-, m+\chi^+j^-, T)\varpi^-(j^-)\dd j^--H(t, \omega^+, \omega^-; \lambda^+, \lambda^-, m, T)\right].
    \end{align*}
    By the first martingale condition of process $\left\{M_t\right\}_{t \geq 0}$, Eq.(\ref{eq:M1}), we simplify the above equation:
    \begin{align}
        0 &= \partial_t H \nonumber\\
        &+ \left(\mu-\frac{\sigma^2}{2}-\lambda^+\E\left(e^{J^+}-1\right)-\lambda^-\E\left(e^{J^-}-1\right)\right)\partial_x H+\frac{\sigma^2}{2}\partial^2_x H \nonumber\\
        &-\left(\kappa^+_1(\xi^+)\kappa^+\lambda^++\kappa^-_1(\xi^-)\kappa^-\lambda^-+\frac{1}{2}\varphi^2(t)\right)\partial_m H
         +\frac{1}{2}\varphi^2(t)\partial^2_m H 
         -\sigma\varphi(t)\partial^2_{xm}\nonumber\\ 
        &+\kappa^+(\theta^+-\lambda^+)\partial_+ H
         +\kappa^-(\theta^--\lambda^-)\partial_- H\nonumber\\
        &+\lambda^+ \left[\int_\mathbb{R}H(t, \omega^+, \omega^-; \lambda^++\beta_{11}j^+, \lambda^-+\beta_{21}j^+, m+\chi^+j^+, T)\varpi^+(j^+)\dd j^+-H(t, \omega^+, \omega^-; \lambda^+, \lambda^-, m, T)\right]\nonumber\\
        &+\lambda^- \left[\int_\mathbb{R}H(t, \omega^+, \omega^-; \lambda^++\beta_{12}j^-, \lambda^-+\beta_{22}j^-, m+\chi^-j^-, T)\varpi^-(j^-)\dd j^--H(t, \omega^+, \omega^-; \lambda^+, \lambda^-, m, T)\right]. \label{eq:drift_MGF_x_Q}
    \end{align}
    Next, we write the form of $H$ such that it is comparable to the MGF under measure $\mathbb{P}$ stated in Proposition \ref{prop:MGF}:
    \begin{align*}
        H(t;\omega, x, \lambda^+, \lambda^-, m, T) = \exp\left(A(t,T)+B(t,T)x+C(t,T)\mathcal{L}^{(+)}(-\chi^+)\lambda^++D(t,T)\mathcal{L}^{(-)}(-\chi^-)+E(t,T)m\right),
    \end{align*} 
    where functions $A$, $B$, $C$, and $D$ are defined in Proposition \ref{prop:MGF} with terminal conditions $A(T,T)=0$, $B(T,T)=\omega$, and $C(T,T)=D(T,T)=0$; $E$ is a time-dependent function with terminal condition $E(T,T)=1$.
    Inserting the partial derivatives of $H$ into Eq.(\ref{eq:drift_MGF_x_Q})
    \begin{align*}
        0 & =\partial_t A+ \partial_t B x +\partial_t C \mathcal{L}^{(+)}\left(-\chi^{+}\right) \lambda^{+}+\partial_t D \mathcal{L}^{(-)}\left(-\chi^{-}\right) \lambda^{-}+\partial_t E m \\
        &+\left(\mu-\frac{\sigma^2}{2}-\lambda^+\E\left(e^{J^+}-1\right)-\lambda^-\E\left(e^{J^-}-1\right)\right)B+\frac{\sigma^2}{2}B^2\\
        & -\left(\kappa_1^{+}\left(\xi^{+}\right) \kappa^{+} \lambda^{+}+\kappa_1^{-}\left(\xi^{-}\right) \kappa^{-} \lambda^{-}+\frac{1}{2} \varphi^2(t)\right) E+\frac{1}{2} \varphi^2(t) E^2 \\
        & +\kappa^{+}\left(\theta^{+}-\lambda^{+}\right) \mathcal{L}^{(+)}\left(-\chi^{+}\right) C+\kappa^{-}\left(\theta^{-}-\lambda^{-}\right) \mathcal{L}^{(-)}\left(-\chi^{-}\right) D \\
        &-BE\sigma\varphi(t)\\
        & +\lambda^{+}\left[\mathcal{L}^{(+)}\left(-B-C \mathcal{L}^{(+)}\left(-\chi^{+}\right) \beta_{11}-D \mathcal{L}^{(-)}\left(-\chi^{-}\right) \beta_{21}-E \chi^{+}\right)-1\right] \\
        & +\lambda^{-}\left[\mathcal{L}^{(-)}\left(-B-C \mathcal{L}^{(+)}\left(-\chi^{+}\right) \beta_{12}-D \mathcal{L}^{(-)}\left(-\chi^{-}\right) \beta_{22}-E \chi^{-}\right)-1\right]
    \end{align*}
    Group terms by state variables
    \begin{align*}
        0 & =\partial_t A+ \mu B - \frac{\sigma^2}{2}(B^2-B)+ \frac{1}{2}\varphi^2(t)(E^2-E)+\kappa^+\theta^+\mathcal{L}^{(+)}(-\chi^+)C+\kappa^-\theta^-\mathcal{L}^{(-)}(-\chi^-)D - BE\sigma\varphi(t) \\
        &+\partial_t B x 
        +\partial_t E m \\
        & +\lambda^{+}\Big[\mathcal{L}^{(+)}\left(-B-C \mathcal{L}^{(+)}\left(-\chi^{+}\right) \beta_{11}-D \mathcal{L}^{(-)}\left(-\chi^{-}\right) \beta_{21}-E \chi^{+}\right)-1\\
        &\phantom{+\lambda^{+}\Big[} +\partial_t C\mathcal{L}^{(+)}(-\chi^+)
        -\E\left(e^{J^+}-1\right)B-\kappa^+_1(\xi^+)\kappa^+E-\kappa^+\mathcal{L}^{(+)}(-\chi^+)C
        \Big] \\
        & +\lambda^{-}\Big[\mathcal{L}^{(-)}\left(-B-C \mathcal{L}^{(+)}\left(-\chi^{+}\right) \beta_{12}-D \mathcal{L}^{(-)}\left(-\chi^{-}\right) \beta_{22}-E \chi^{-}\right)-1\\
        &\phantom{+\lambda^{+}\Big[} +\partial_t D\mathcal{L}^{(-)}(-\chi^-)
        -\E\left(e^{J^+}-1\right)B-\kappa^-_1(\xi^-)\kappa^-E-\kappa^-\mathcal{L}^{(-)}(-\chi^-)D
        \Big].
    \end{align*}
    Since the above equation holds for all value of state variables, $x,\ \lambda^+,\ \lambda^-$, we infer for $t \leq T$, $\partial_t B(t,T)=\partial_t E(t,T)=0$, and thus $B(t,T)=\omega$ and $E(t,T)=1$.
    The above equation becomes
    \begin{align*}
        0 & =\partial_t A+ \mu B - \frac{\sigma^2}{2}(\omega^2-\omega)+\kappa^+\theta^+\mathcal{L}^{(+)}(-\chi^+)C+\kappa^-\theta^-\mathcal{L}^{(-)}(-\chi^-)D - \omega\sigma\varphi(t) \\
        & +\lambda^{+}\Big[\mathcal{L}^{(+)}\left(-\omega-C \mathcal{L}^{(+)}\left(-\chi^{+}\right) \beta_{11}-D \mathcal{L}^{(-)}\left(-\chi^{-}\right) \beta_{21}- \chi^{+}\right)-1\\
        &\phantom{+\lambda^{+}\Big[} +\partial_t C\mathcal{L}^{(+)}(-\chi^+)
        -\E\left(e^{J^+}-1\right)\omega-\kappa^+_1(\xi^+)\kappa^+-\kappa^+\mathcal{L}^{(+)}(-\chi^+)C
        \Big] \\
        & +\lambda^{-}\Big[\mathcal{L}^{(-)}\left(-\omega-C \mathcal{L}^{(+)}\left(-\chi^{+}\right) \beta_{12}-D \mathcal{L}^{(-)}\left(-\chi^{-}\right) \beta_{22}- \chi^{-}\right)-1\\
        &\phantom{+\lambda^{+}\Big[} +\partial_t D\mathcal{L}^{(-)}(-\chi^-)
        -\E\left(e^{J^+}-1\right)\omega-\kappa^-_1(\xi^-)\kappa^--\kappa^-\mathcal{L}^{(-)}(-\chi^-)D
        \Big].
    \end{align*}
    Recall we choose the $\mathcal{F}_t$-adapted process as $\varphi(t)=\varphi+\varphi^+\lambda^+_t+\varphi^-\lambda^-_t$, so expanding the $\varphi(t)$ term in the above equation gives
    \begin{align*}
        0 & =\partial_t A+ \mu \omega + \frac{\sigma^2}{2}(\omega^2-\omega)+\kappa^+\theta^+\mathcal{L}^{(+)}(-\chi^+)C+\kappa^-\theta^-\mathcal{L}^{(-)}(-\chi^-)D - \omega\sigma\varphi\\
        & +\lambda^{+}\big[\mathcal{L}^{(+)}\left(-\omega-C \mathcal{L}^{(+)}\left(-\chi^{+}\right) \beta_{11}-D \mathcal{L}^{(-)}\left(-\chi^{-}\right) \beta_{21}-\chi^{+}\right)-1\\
        &\phantom{+\lambda^{+}\Big[}+\partial_tC\mathcal{L}^{(+)}(-\chi^+)
        -\E\left(e^{J^+}-1\right)\omega-\kappa^+_1(\xi^+)\kappa^+-\kappa^+\mathcal{L}^{(+)}(-\chi^+)C
        - \omega\sigma\varphi^+
        \big] \\
        & +\lambda^{-}\big[\mathcal{L}^{(-)}\left(-\omega-C \mathcal{L}^{(+)}\left(-\chi^{+}\right) \beta_{12}-D \mathcal{L}^{(-)}\left(-\chi^{-}\right) \beta_{22}-\chi^{-}\right)-1\\
        &\phantom{+\lambda^{+}\Big[}+\partial_tD\mathcal{L}^{(-)}(-\chi^-)
        -\E\left(e^{J^+}-1\right)\omega-\kappa^-_1(\xi^-)\kappa^--\kappa^-\mathcal{L}^{(-)}(-\chi^-)D
        - \omega\sigma\varphi^-
        \big]
    \end{align*}
    The above equation holds for all values of state variables. Together with the parameter changes stated in Proposition \ref{prop:HawkesUnderQ}, we have    
    \begin{align*}
        \partial_t A &= -(\mu - \sigma\varphi)\omega - \frac{\sigma^2}{2}(\omega^2-\omega)-\kappa^+\theta^{+,\mathbb{Q}}C-\kappa^-\theta^{-,\mathbb{Q}}D\\
        \partial_t C &= -\mathcal{L}^{+, \mathbb{Q}}(-\omega-C\beta^\mathbb{Q}_{11}-D\beta^\mathbb{Q}_{21}-\chi^+)+1
        +\frac{\E\left(e^{J^+}-1\right)+\sigma \varphi^+}{\mathcal{L}^{(+)}(-\chi^+)}\omega+\kappa^+C\\
        \partial_t D &= -\mathcal{L}^{-, \mathbb{Q}}(-\omega-C\beta^\mathbb{Q}_{21}-D\beta^\mathbb{Q}_{22}-\chi^-)+1
        +\frac{\E\left(e^{J^-}-1\right)+\sigma \varphi^-}{\mathcal{L}^{(-)}(-\chi^-)}\omega+\kappa^-D
    \end{align*}
    By choosing 
    \begin{align*}
      \varphi   &= \sigma^{-1}(\mu-r),\\
      \varphi^+ &= \sigma^{-1}\left[\mathcal{L}^{(+)}(-\chi^+) \E^\mathbb{Q}\left(e^{J^+}-1\right)-\E\left(e^{J^+}-1\right)\right],\\
      \varphi^- &= \sigma^{-1}\left[\mathcal{L}^{(+)}(-\chi^+) \E^\mathbb{Q}\left(e^{J^-}-1\right)-\E\left(e^{J^-}-1\right)\right],
    \end{align*}
     the MGF of $\left\{X_t\right\}_{t \geq 0}$ under measure $\mathbb{Q}$ can be written as
    \begin{align*}
        \E^\mathbb{Q}\left(e^{\omega X_T}\big|\mathcal{F}_t\right)
        &= \exp\left(A(t,T)+\omega x + C(t,T)\mathcal{L}^{(+)}(-\chi^+)\lambda^{+}_t  + D(t,T)\mathcal{L}^{(-)}(-\chi^-)\lambda^{-}_t\right)\\ 
        &= \exp\left(A(t,T)+\omega x + C(t,T)\lambda^{+, \mathbb{Q}}_t +D(t,T)\lambda^{-, \mathbb{Q}}_t \right), 
    \end{align*}
    where the functions $A$, $C$, and $D$ solve the following system of PDEs:
    \begin{align*}
        \partial_t A(t,T) &= -r\omega - \frac{\sigma^2}{2}(\omega^2-\omega)-\kappa^+\theta^{+,\mathbb{Q}}C(t,T)-\kappa^-\theta^{-,\mathbb{Q}}D(t,T)\\
        \partial_t C(t,T) &= -\mathcal{L}^{+, \mathbb{Q}}(-\omega-C(t,T)\beta^\mathbb{Q}_{11}-D(t,T)\beta^\mathbb{Q}_{21}-\chi^+)+1
        + \E^\mathbb{Q}\left(e^{J^+}-1\right)\omega+\kappa^+C(t,T)\\
        \partial_t D(t,T) &= -\mathcal{L}^{-, \mathbb{Q}}(-\omega-C(t,T)\beta^\mathbb{Q}_{21}-D(t,T)\beta^\mathbb{Q}_{22}-\chi^-)+1
        + \E^\mathbb{Q}\left(e^{J^-}-1\right)\omega+\kappa^-D(t,T).
    \end{align*}
     Finally, by comparing the above MGF under $\mathbb{Q}$ to that under $\mathbb{P}$ stated in Proposition \ref{prop:MGF}, we conclude that the dynamics of the log returns $\left\{X_t\right\}_{t \geq 0}$ under measure $\mathbb{Q}$ can be retrieved by the substitution of model parameters stated in Eq.\ref{eq:mgf_q1}.

\subsection{Derivation of the likelihood}
Suppose we observe $k-1$ jump events and collect the following:
\begin{itemize}
\item positive and negative jump event times $\mathcal{T}^+ = \left\{T^{+}_{[1]}, T^{+}_{[2]},..., T^{+}_m\right\}$ and $\mathcal{T}^- =\left\{T^{-}_{[1]}, T^{-}_{[2]} ..., T^{-}_{[n]}\right\}$ respectively,
\item the corresponding jump sizes of positive and negative jump events $\left\{J^+(T^{+}_{[1]}), J^+(T^{+}_{[2]}),..., J^+(T^{+}_{[m]})\right\}$ and $\left\{J^-(T^{-}_{[1]}), J^-(T^{-}_{[2]}), ..., J^-(T^{-}_{[n]})\right\}$ respectively, and,
\item a set of ordered set of union jump times $\mathcal{T}^\pm = \mathcal{T}^+ \cup \mathcal{T}^- = \left\{T^\pm_{[1]}, T^\pm_{[2]},..., T^\pm_{[k-1]}\right\} $, 
\end{itemize}
 the hazard rate of observing the $k^{\text{th}}$ jump between time $T^\pm_{[k-1]}$ and $s > T^\pm_{[k-1]}$ is 
\begin{align*}
F^\pm \big(s\big|T^\pm_{[k-1]}\big) &= 1-\exp\left(-\int_{T^\pm_{[k-1]}}^s(\lambda^+_t + \lambda^-_t)\dd t\right)\\
                                    &= 1-\exp\left(-I^+(T^\pm_{[k-1]}, s)-I^-(T^\pm_{[k-1]}, s)\right),
\end{align*}
where 
\begin{align*}
  I^+(T^\pm_{[k-1]}, s) & = \int_{T^\pm_{[k-1]}}^s \lambda^+_t \dd t = \theta^+(s-T^\pm_{[k-1]}) + (\lambda^+(T^\pm_{[k-1]})-\theta^+)\frac{1-e^{-\kappa^+(s-T^\pm_{[k-1]})}}{\kappa^+}\\
  I^-(T^\pm_{[k-1]}, s) & = \int_{T^\pm_{[k-1]}}^s \lambda^-_t \dd t = \theta^+(s-T^\pm_{[k-1]}) + (\lambda^-(T^\pm_{[k-1]})-\theta^+)\frac{1-e^{-\kappa^-(s-T^\pm_{[k-1]})}}{\kappa^-}.
\end{align*}
The pds of the inter-arrival time given the $(k-1)^\text{th}$ jump is 
\begin{align*}
f^\pm  \big(s\big|T^\pm_{[k-1]}\big) = (\lambda^+_s + \lambda^-_s)\exp(-I^+-I^-).
\end{align*}
Suppose the next jump happens at time $T^\pm_{[k]}$, the probability of this jump being a positive jump is 
\begin{align*}
\mathbb{P}\left(\dd N^{(1)}(s)=1\big|s=T^\pm_{[k]}\right) = \frac{\lambda^+(T^\pm_{[k]}-)}{\lambda^+(T^\pm_{[k]}-)+\lambda^-(T^\pm_{[k]}-)},
\end{align*}
where $\dd N^{(1)}(s) = N^{(1)}(s) - N^{(1)}(s-)$;
And similarly for negative jump
\begin{align*}
  \mathbb{P}\left(\dd N^{(2)}(s)=1\big|s=T^\pm_{[k]}\right) = \frac{\lambda^-(T^\pm_{[k]}-)}{\lambda^+(T^\pm_{[k]}-)+\lambda^-(T^\pm_{[k]}-)}.
  \end{align*}
Accordingly, the pdf of a positive jump occurring at time $s$ is 
\begin{align*}
  f^+  \big(s\big|T^\pm_{[k-1]}\big) = f^\pm\big(s\big|T^\pm_{[k-1]}\big)\mathbb{P}\left(\dd N^{(1)}(s)=1\big| s=T^\pm_{[k]}\right) = \lambda^+_s\exp(-I^+-I^-),
  \end{align*}
  and similarly for that of a negative jump
  \begin{align*}
    f^-  \big(s\big|T^\pm_{[k-1]}\big) = \lambda^-_s\exp(-I^+-I^-).
    \end{align*}
Therefore, the likelihood function of $\left\{\lambda^+_t\right\}_{t \leq T^\pm_{[k-1]}}$, $\left\{\lambda^-_t\right\}_{t \leq T^\pm_{[k-1]}}$, 
and the jumps sizes densities $\varpi^+$ and $\varpi^-$ to the observed inter-arrival times is 
\begin{align*}
&\mathbb{P}\left(\mathcal{T}^+, \mathcal{T}^-\big| \left\{\lambda^+_t\right\}_{t \leq T^\pm_{[k-1]}}, \left\{\lambda^-_t\right\}_{t \leq T^\pm_{[k-1]}}, \varpi^+, \varpi^-\right)\\
&= \prod_{j=1}^m \lambda^+(T^+_{[m]})\varpi^+(J^-(T^+_{[m]})) \\ 
&\times \prod_{j=1}^n \lambda^-(T^-_{[n]})\varpi^-(J^-(T^-_{[n]}))\\
&\times \prod_{j=1}^{k-1}\exp\left(-I^+(T^\pm_{[j]}, T^\pm_{[j-1]}) - I^-(T^\pm_{[j]}, T^\pm_{[j-1]})\right).
\end{align*}

\section{Goodness of fit assessment by Q-Q plot}\label{sec:QQplot}
Our assessment relies on the work by \cite{watanabe1964discontinuous} who states that a Poisson process can be characterized by the form of its compensator, $\Lambda(T)=\int_0^T\lambda(t)\text{d}t$.
 This characterisation is later on known as the random time change theorem, see Section 7.4 of \citep{daley2003introduction} and reference therein.
 We rewrite the multivariate random time change theorem stated in Section 9.3 of \cite{laub2021elements} in a bidimentional version to suit our purpose here. 
 \begin{proposition}
Bidimensional random time change. (Theorem 9.3 of \cite{laub2021elements}) \\
Suppose we observe two counting processes, $N^{(1)}(t)$ and $N^{(2)}(t)$, which are characterised by intensities $\lambda^{+}_t$ and $\lambda^{-}_t$, respectively.
 The event arrival times of the two counting processes are labelled as 
  $\left\{T^{(1)}_1,\ T^{(1)}_2,...\right\}\text{ and } \left\{T^{(2)}_1,\ T^{(2)}_2,...\right\} $.
 In addition, suppose $\lambda_t = \lambda^{+}_t +\lambda^{-}_t$ is positive over $[0,\ T]$ and 
 $\int_0^T \lambda(t) \text{d}t < \infty$ almost surely, then the transformed event arrival times
 \begin{equation*}
    \left\{\int_0^{T^{(1)}_1}\lambda^+_t \text{d}t,\ \int_{0}^{T^{(1)}_2}\lambda^+_t \text{d}t,...\right\} \text{ and }
    \left\{\int_0^{T^{(2)}_1}\lambda^-_t \text{d}t,\ \int_{0}^{T^{(2)}_2}\lambda^-_t \text{d}t,...\right\}
 \end{equation*}
 are events arrival times of two Poisson processes with unit rate.
\end{proposition}

In order words, the transformed events inter-arrival times, 
\begin{align*}
  &\left\{\int_0^{T^{(1)}_1}\lambda^+_t \text{d}t,\ \int_{T^{(1)}_1}^{T^{(1)}_2}\lambda^+_t \text{d}t,\ \int_{T^{(1)}_2}^{T^{(1)}_3}\lambda^+_t \text{d}t,...\right\} \text{ and }\\
  &\left\{\int_0^{T^{(2)}_1}\lambda^-_t \text{d}t,\ \int_{T^{(2)}_1}^{T^{(2)}_2}\lambda^-_t \text{d}t,\ \int_{T^{(2)}_2}^{T^{(2)}_3}\lambda^-_t \text{d}t,...\right\}
\end{align*}
are two random variables that follows exponential distribution with unit rate. 
Note that the event of $N^{(2)}$ can arrive between two events of $N^{(1)}$ and vice versa, e.g. it is possible that 
$T^{(1)}_{k} < T^{(2)}_{j} < T^{(1)}_{k+1}$ for some integers $k$ and $j$.  
The random time change theorem enables us to perform residual analysis to assess the goodness-of-fit of time-inhomogeneous Poisson processes. 

\begin{theorem}
Consider an unbounded, increasing sequence of time points $\{t_1,\ t_2,\ ...\}$ in $(0,\ \infty)$, and a monotone, continuous function 
$\Lambda(\cdot)$ such that $\lim_{t \rightarrow \infty}\Lambda(t)=\infty$ almost surely. 
The transformed sequence $\{\Lambda(t_1),\ \Lambda(t_2),\ ...\}$ is a realisation of a unit rate Poisson process if and only if the original sequence $\{t_1,\ t_2,\ ...\}$
is a realisation from the point process characterised by $\Lambda(\cdot)$.
\end{theorem}

\begin{proof}
  See \cite{brown1988simple}.
\end{proof}
%
%
We refer readers to Section 9 of \cite{laub2021elements} for further details regarding the application of random time change theorem on Hawkes processes.

\end{document}